\documentclass[journal]{IEEEtran}
\IEEEoverridecommandlockouts
\usepackage{lineno}
\usepackage{hyperref}
\usepackage{cite}
\usepackage{amsmath,amssymb,amsfonts,cases} 
\usepackage{amsmath}
\usepackage{amsthm}

\usepackage{graphicx}
\usepackage{textcomp}
\usepackage{xcolor}
\usepackage{graphicx}
\usepackage{float}
\usepackage{subfigure}
\usepackage{amsmath,mleftright}
\usepackage{amsfonts,amssymb}
\usepackage{mathrsfs}
\usepackage{mathtools}
\usepackage{algorithm}
\usepackage{algorithmic}
\usepackage{bm}
\usepackage{multirow}
\usepackage{array}
\usepackage{amssymb}
\usepackage{amsmath}
\usepackage{cite}
\usepackage{url}
\usepackage{xcolor}
\usepackage{cite,graphicx,amsmath,amssymb}
\usepackage{subfigure}
\usepackage{fancyhdr}
\usepackage{mdwmath}
\usepackage{mdwtab}
\usepackage{caption}
\usepackage{amsthm}
\usepackage{setspace}
\usepackage{bm}
\usepackage{mathtools}
\usepackage{dsfont}
\usepackage{bbm}
\usepackage{framed}
\newtheorem{remark}{Remark}
\newtheorem{theorem}{Theorem}

\newtheorem{lemma}{Lemma}

\newtheorem{corollary}{Corollary}

\makeatletter
\newcommand{\biggg}{\bBigg@{3}}
\newcommand{\Biggg}{\bBigg@{3.5}}
\makeatother
\makeatletter
\renewcommand{\maketag@@@}[1]{\hbox{\m@th\normalsize\normalfont#1}}%
\makeatother
\def\BibTeX{{\rm B\kern-.05em{\sc i\kern-.025em b}\kern-.08em
    T\kern-.1667em\lower.7ex\hbox{E}\kern-.125emX}}
    \expandafter\def\expandafter\normalsize\expandafter{%
    \normalsize%
    \setlength\abovedisplayskip{4pt}%
    \setlength\belowdisplayskip{4pt}%
    \setlength\abovedisplayshortskip{2pt}%
    \setlength\belowdisplayshortskip{2pt}%
}
\allowdisplaybreaks[4]
\begin{document}
\title{On the Maintainability of Pinching-Antenna Systems: A Failure-Repair Perspective}
\author{Chongjun~Ouyang, Hao~Jiang, Zhaolin~Wang, Yuanwei~Liu, and Zhiguo~Ding\vspace{-10pt}
\thanks{C. Ouyang and H. Jiang are with the School of Electronic Engineering and Computer Science, Queen Mary University of London, London, E1 4NS, U.K. (e-mail: \{c.ouyang, hao.jiang\}@qmul.ac.uk).}
\thanks{Z. Wang and Y. Liu are with the Department of Electrical and Electronic Engineering, The University of Hong Kong, Hong Kong (email: \{zhaolin.wang, yuanwei\}@hku.hk).}
\thanks{Z. Ding is with the School of Electrical and Electronic Engineering (EEE), Nanyang Technological University, Singapore 639798 (zhiguo.ding@ntu.edu.sg).}}
\maketitle
\begin{abstract}
The pinching-antenna system (PASS) enables wireless channel reconfiguration through optimized placement of pinching antennas along dielectric waveguides. In this article, a unified analytical framework is proposed to characterize the maintainability of PASS. Within this framework, random waveguide failures and repairs are modeled by treating the waveguide lifetime and repair time as exponentially distributed random variables, which are characterized by the failure rate and the repair rate, respectively. The operational state of the waveguide is described by a two-state continuous-time Markov chain, for which the transition probabilities and steady-state probabilities of the waveguide being \emph{working} or \emph{failed} are analyzed. By incorporating the randomness of the waveguide operational state into the transmission rate, system maintainability is characterized using the probability of non-zero rate (PNR) and outage probability (OP). The proposed framework is applied to both a conventional PASS employing a single long waveguide and a segmented waveguide-enabled pinching-antenna system (SWAN) composed of multiple short waveguide segments under two operational protocols: segment switching (SS) and segment aggregation (SA). Closed-form expressions for the PNR and OP are derived for both architectures, and the corresponding scaling laws are analyzed with respect to the service-region size and the number of segments. It is proven that both SS-based and SA-based SWAN achieve higher PNR and lower OP than conventional PASS, which confirms the maintainability advantage of segmentation. Numerical results demonstrate that: \romannumeral1) the maintainability gain of SWAN over conventional PASS increases with the number of segments, and \romannumeral2) SA provides stronger maintainability than SS.
\end{abstract}
\section{Introduction}
Reconfigurable antennas have recently emerged as a pivotal building block for tri-hybrid multiple-input multiple-output (MIMO) architectures \cite{heath2025tri}. These architectures are expected to meet the growing demand for high-capacity connectivity and high-precision sensing in sixth-generation (6G) networks, and they may reshape future wireless transmission paradigms \cite{heath2025tri}. Reconfigurable antennas adjust their radiation patterns in real time, which enables flexible electromagnetic (EM) beam control and channel reconfiguration. Representative examples include reconfigurable intelligent surfaces (RISs) \cite{wu2025intelligent}, fluid antennas \cite{new2024tutorial}, and movable antennas \cite{zhu2025tutorial}. RISs manipulate EM waves through programmable meta-atoms \cite{wu2025intelligent}, whereas fluid and movable antennas change their physical positions to mitigate small-scale fading \cite{new2024tutorial,zhu2025tutorial}. 

Despite their strong potential in communications and sensing, the reconfiguration capability of these technologies is typically limited to apertures that span only a few to several tens of wavelengths. This spatial extent does not sufficiently address two fundamental constraints imposed by wave propagation physics, particularly at the high-frequency bands envisioned for 6G systems, such as the 7-24 GHz upper mid-bands \cite{bjornson2025enabling}. These constraints include severe large-scale free-space path loss and high sensitivity to signal blockage \cite{liu2025pinching}.
\subsection{Pinching Antennas}
To address these limitations, NTT DOCOMO introduced the pinching-antenna system (PASS) and demonstrated a working prototype \cite{suzuki2022pinching}. PASS uses low-loss dielectric waveguides as signal conduits, where small dielectric particles, termed pinching antennas (PAs), are deployed along the waveguides. Each PA can radiate EM signals into free space or couple free-space signals back into the waveguide. When the PA locations are adjusted, the phases and amplitudes of the radiated signals can be controlled, which enables flexible beam-pattern shaping via pinching beamforming \cite{liu2025pinching}. In contrast to existing reconfigurable-antenna technologies, PASS supports waveguides with arbitrarily long extent. This feature allows PAs to be placed close to users, which creates strong line-of-sight (LoS) links. PASS therefore alleviates large-scale path loss and reduces vulnerability to blockage \cite{suzuki2022pinching}. Overall, PASS combines low propagation loss with a high degree of EM-field controllability, which makes it a promising architecture for next-generation wireless communications \cite{liu2025pinching}.

Early field tests by NTT DOCOMO confirmed the feasibility of using pinched waveguides to enhance network coverage and throughput \cite{yamamoto2021pinching}. These demonstrations stimulated growing research interest in PASS for communications enhancement. The work in \cite{ding2024flexible} analyzed the average transmission rate when PAs serve mobile users and provided a theoretical comparison with conventional fixed-position antenna systems. Subsequent studies analyzed outage probability (OP) \cite{tyrovolas2025performance}, array gain \cite{ouyang2025array}, channel capacity \cite{ouyang2025capacity}, and the impact of LoS blockage \cite{ding2025blockage} of PASS channels. Collectively, these results indicate that PASS alleviates large-scale path loss and LoS blockage, and it can outperform traditional reconfigurable-antenna systems, such as RISs and fluid/movable antennas \cite{samy2025pinching,ouyang2025array}. 

Building on this foundation, a range of pinching beamforming algorithms have been developed to optimize the placement of active PAs along the waveguide. Representative designs have been conducted for single-user PASS \cite{xu2024rate} and for multiuser PASS under orthogonal multiple access \cite{zeng2025energy} and non-orthogonal multiple access \cite{bereyhi2025downlink,wang2025modeling}. Beyond communications, PASS has also attracted increasing interest in wireless sensing. The Cram\'{e}r-Rao bounds (CRBs) for PASS-assisted localization were analyzed in \cite{ding2025pinching,jiang2025pinching}. Other sensing applications, such as target detection \cite{qin2025joint}, target tracking \cite{wang2025wireless}, and radar cross-section sensing \cite{ouyang2025isac}, have also been studied.
\subsection{From Monolithic Waveguides to Segmented Waveguides}

\begin{figure}[!t]
\centering
    \subfigure[System setup.]
    {
        \includegraphics[width=0.45\textwidth]{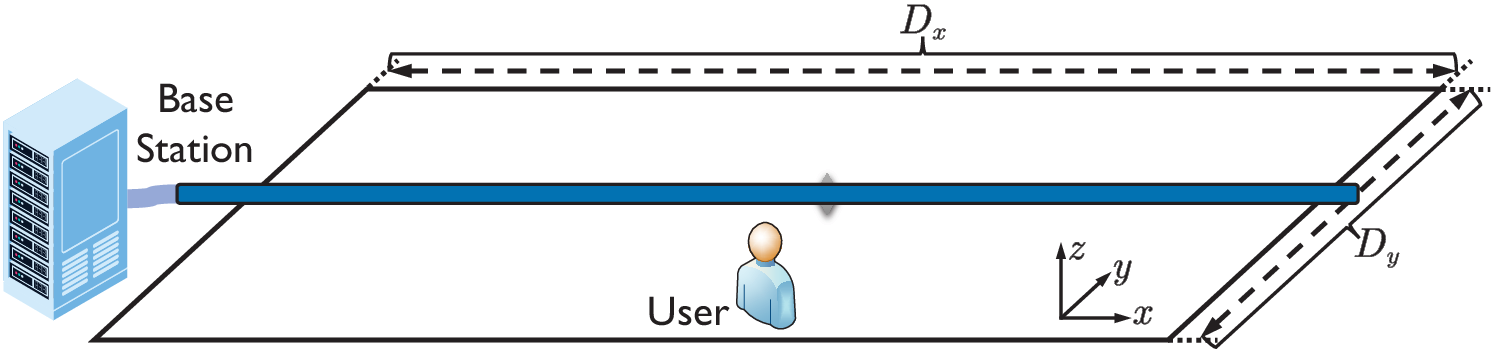}
	   \label{Figure: PASS_System_Model}
    }
   \subfigure[Monolithic waveguide.]
    {
        \includegraphics[width=0.45\textwidth]{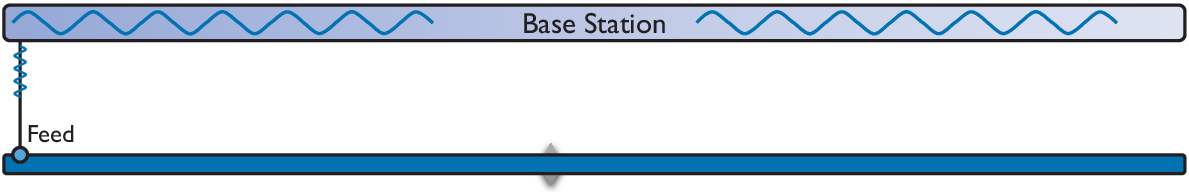}
	   \label{Figure: PASS_System_Model1}
    }
\caption{Illustration of conventional PASS.}
\label{Figure: Conventional_PASS_System_Model}
\vspace{-10pt}
\end{figure}

The studies reviewed above establish an initial theoretical foundation for PASS and demonstrate its advantages in low propagation loss and high flexibility for improving communication throughput and sensing accuracy. Despite this progress, a substantial gap remains between theoretical analysis and practical deployment. The key advantage of PASS is its ability to facilitate stable LoS links through long monolithic dielectric waveguides, as illustrated in {\figurename} {\ref{Figure: Conventional_PASS_System_Model}}. In principle, such waveguides can be fabricated and extended to cover large service regions. However, a single long waveguide introduces three fundamental challenges. First, in-waveguide propagation loss increases with waveguide length \cite{xu2025pinching}. This loss can significantly degrade performance and may offset the gains obtained from reduced free-space path loss. Second, long-waveguide deployments raise serious maintainability concerns. A failure that occurs at an arbitrary location is difficult to localize, and the repair process often requires replacement of the entire structure. This feature results in high maintenance cost and limited scalability. Third, and most importantly, monolithic-waveguide PASS does not admit a tractable uplink signal model. PAs couple signals into the waveguide in a passive manner. During propagation toward the feed point, the signal collected by one antenna can re-radiate through other antennas. This inter-antenna radiation (IAR) effect complicates the uplink analysis and undermines tractability. As a result, most existing studies either neglect IAR in multi-PA settings \cite{tegos2024minimum} or focus on single-PA deployments \cite{zeng2025sum,zeng2025energyuplink}, which yields simplified models with limited fidelity.

To address these challenges, a segmented waveguide-enabled pinching-antenna system (SWAN) was proposed in \cite{ouyang2025uplink}. Unlike conventional PASS, which relies on a single monolithic waveguide, SWAN employs multiple short dielectric waveguide segments that are arranged end-to-end. These segments are not physically interconnected. Instead, each segment has an independent feed point for signal injection and extraction. The segment signals are forwarded to the base station (BS) through wired connections, such as optical fiber or low-loss coaxial cables; see {\figurename} {\ref{Figure: New_SWAN_System_Model}}. By activating at most one PA per segment, SWAN eliminates IAR and yields a tractable uplink multi-PA model. In addition, signal propagation and waveguide maintenance remain confined to individual segments. This feature reduces in-waveguide attenuation and lowers maintenance cost. These advantages have stimulated increasing research interest in SWAN \cite{jiang2025segmented,gu2025sum}.

\begin{figure}[!t]
\centering
    \subfigure[System setup.]
    {
        \includegraphics[width=0.45\textwidth]{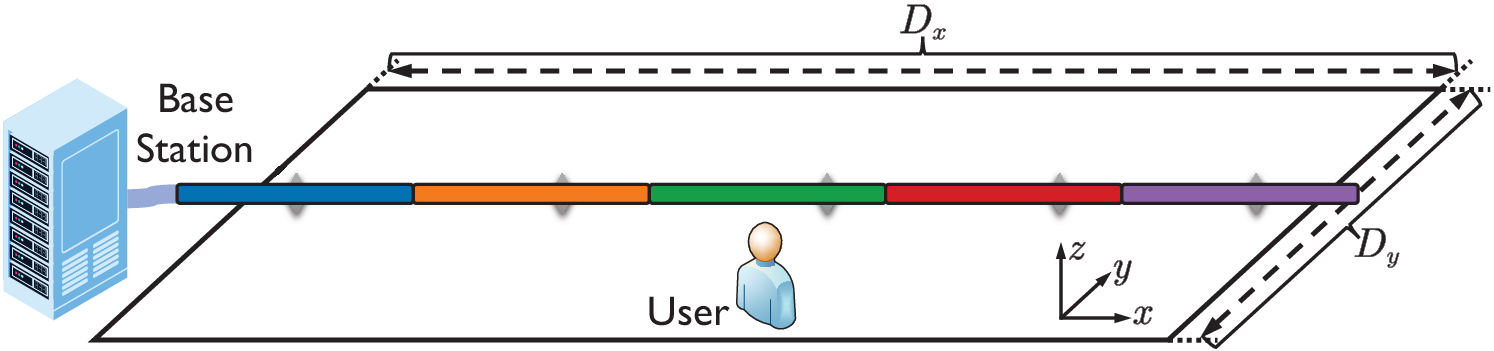}
	   \label{Figure: SWAN_System_Model}
    }
   \subfigure[Segmented waveguide.]
    {
        \includegraphics[width=0.45\textwidth]{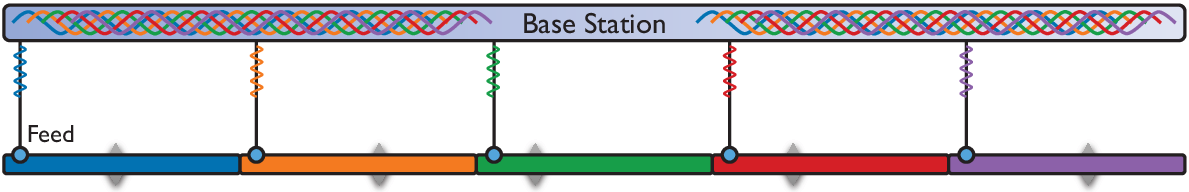}
	   \label{Figure: SWAN_System_Model1}
    }
\caption{Illustration of SWAN.}
\label{Figure: New_SWAN_System_Model}
\vspace{-10pt}
\end{figure}

\subsection{Motivation and Contributions}
Existing works have primarily emphasized the communication-theoretic benefits of SWAN, with particular attention to reduced in-waveguide propagation loss and the suppression of IAR \cite{ouyang2025uplink,jiang2025segmented,gu2025sum}. However, the potential of SWAN to improve system maintainability and reliability has so far been addressed mainly through qualitative arguments, and it lacks rigorous theoretical quantification. More broadly, an analytical framework that captures how maintainability affects the performance of PASS-based communications is still missing. In practical deployments, waveguide lifetimes are shaped by random failures that arise from hardware defects, material aging, and environmental conditions. Repair times also vary due to logistical constraints, technician availability, and replacement delays. These factors determine the long-term availability and reliability of PASS, and they can strongly influence communication performance. Without a rigorous analytical framework, it is difficult to quantify these effects or to assess the maintainability advantages of SWAN over conventional PASS.

These observations motivate this work. The objective is to establish a general analytical framework that captures how system maintainability affects PASS-based communications, and to use this framework to quantify the maintainability advantages of SWAN over conventional PASS in a rigorous manner. The main contributions are summarized as follows.
\begin{itemize}
\item We propose a unified analytical framework for maintainability and reliability evaluation in PASS. Specifically, we model the lifetime and repair time of a dielectric waveguide as exponentially distributed random variables, which are parameterized by a failure rate and a repair rate per unit waveguide length, respectively. Under this model, we describe the operational state of the waveguide by a two-state continuous-time Markov chain (CTMC) with a \emph{working} state and a \emph{failed} state. We derive the time-dependent transition probabilities of the CTMC and the steady-state probabilities that the waveguide is \emph{working} or \emph{failed}. We then incorporate the randomness of waveguide availability into the received signal-to-noise ratio (SNR), which yields a binary random transmission rate model. On this basis, we employ the probability of non-zero rate (PNR) and the OP to characterize maintainability and reliability in PASS-based communications.
\item We apply the proposed framework to conventional PASS that uses a single monolithic waveguide. We derive closed-form expressions for the achievable PNR and OP. To provide further insight, we characterize the scaling laws of the PNR and OP with respect to (w.r.t.) the waveguide length. The analysis shows that the PNR decays quadratically as the waveguide length increases, which reveals a fundamental maintainability limitation of monolithic-waveguide PASS over large service regions.    
\item We then analyze the maintainability of SWAN under two operational protocols, including \emph{segment switching (SS)} and \emph{segment aggregation (SA)}. For SS, we derive closed-form expressions for both PNR and OP, and show that SS-based SWAN achieves a higher PNR and a lower OP than conventional PASS, which demonstrates its improved maintainability. For SA, we derive a closed-form expression for the PNR, and an analytically tractable upper bound for the OP. We further analyze their scaling laws w.r.t. the number of segments, and prove that SA-based SWAN approaches its performance limits faster than SS-based SWAN, which yields superior maintainability and reliability performance.
\item We present numerical results to validate the theoretical analysis and demonstrate that: \romannumeral1) both SS-based SWAN and SA-based SWAN achieve higher PNR and lower OP than conventional PASS, and the performance gains increase with the service region size; \romannumeral2) SA-based SWAN outperforms SS-based SWAN in terms of both PNR and OP; and \romannumeral3) increasing the number of segments improves PNR and reduces OP under both SS and SA.
\end{itemize}

The remainder of this article is organized as follows. Section {\ref{Section: System Model}} introduces the system model for SWAN. Section \ref{Section: Maintenance of the Conventional PASS} presents an analytical framework for characterizing the maintainability of PASS-based communications. Section \ref{Section:Uplink SWAN} analyzes the maintainability of SWAN. Section \ref{Section_Numerical_Results} provides numerical results. Finally, Section \ref{Section_Conclusion} concludes the paper.
\subsubsection*{Notations}
Throughout this article, scalars and vectors are denoted by non-bold and bold letters, respectively. The transpose is denoted by $(\cdot)^{\mathsf{T}}$. The notations $\lvert a\rvert$ and $\lVert \mathbf{a} \rVert$ represent the magnitude of scalar $a$ and the norm of vector $\mathbf{a}$, respectively. The expectation and variance operators are denoted by ${\mathbbmss{E}}\{\cdot\}$ and ${\mathbbmss{V}}\{\cdot\}$, respectively. The sets $\mathbbmss{C}$ and $\mathbbmss{R}$ denote the complex and real spaces, respectively. The shorthand $[N]$ denotes the set $\{1,\ldots, N\}$. The ceiling operator is denoted by $\lceil\cdot\rceil$. The notation ${\mathcal{CN}}(\mu,\sigma^2)$ refers to a circularly symmetric complex Gaussian distribution with mean $\mu$ and variance $\sigma^2$, and ${\rm{Exp}}(\lambda)$ denotes an exponential distribution with rate parameter $\lambda$. Finally, ${\mathcal{O}}(\cdot)$ and $o(\cdot)$ denote the standard big-O and little-o notations, respectively. 
\section{System Model}\label{Section: System Model}
Consider an uplink system in which a BS uses a pinched dielectric waveguide to serve a single-antenna user, as illustrated in {\figurename} {\ref{Figure: Conventional_PASS_System_Model}}. The user is distributed over a rectangular service region with dimensions $D_x$ and $D_y$ along the $x$- and $y$-axes, respectively. The user location is denoted by ${\mathbf{u}}\triangleq[u_{x},u_{y},0]^{\mathsf{T}}$. The waveguide extends along the $x$-axis and spans the entire horizontal range of the service region. PAs are deployed along the waveguide to receive uplink signals \cite{suzuki2022pinching}. 
\subsection{Conventional PASS}
We first consider the conventional PASS architecture that uses a monolithic waveguide, which is a single long continuous structure, as shown in {\figurename} {\ref{Figure: PASS_System_Model}}. 
\subsubsection{System Architecture}
When multiple PAs are deployed along the waveguide to receive uplink signals, the signal collected by one PA can re-radiate through other PAs during propagation toward the feed point. This IAR effect leads to a mathematically intractable uplink model, and a tractable multi-PA model is not available for analysis. For this reason, we consider a single PA on the monolithic waveguide. The PA location is denoted by ${\bm\psi}^{\rm{M}}\triangleq[\psi^{\rm{M}},0,d]^{\mathsf{T}}$, where $d$ is the deployment height of the waveguide. The feed point lies at the front-left end of the waveguide, as shown in {\figurename} {\ref{Figure: PASS_System_Model1}}. Its location is ${\bm\psi}_0^{\rm{M}}\triangleq[\psi_0^{\rm{M}},0,d]^{\mathsf{T}}$ with $\psi_0^{\rm{M}}\leq \psi^{\rm{M}}$.
\subsubsection{Channel Model}
PASS targets high-frequency bands \cite{suzuki2022pinching}, where LoS propagation typically dominates \cite{ouyang2024primer}. We therefore adopt a free-space LoS model for the spatial channel coefficient between the PA and the user \cite{ouyang2024primer}:
\begin{align}
h_{\rm{o}}({\bm\psi}^{\rm{M}},{\mathbf{u}})\triangleq
\frac{\eta^{\frac{1}{2}}{\rm{e}}^{-{\rm{j}}k_0\lVert{\bm\psi}^{\rm{M}}-{\mathbf{u}}\rVert}}{\lVert{\bm\psi}^{\rm{M}}-{\mathbf{u}}\rVert},
\end{align}
where $\eta\triangleq\frac{c^2}{16\pi^2f_{\rm{c}}^2}$, $f_{\rm{c}}$ is the carrier frequency, $c$ is the speed of light, $k_0\triangleq\frac{2\pi}{\lambda}$ is the wavenumber, and $\lambda$ is the free-space wavelength. The in-waveguide propagation coefficient between the feed point and the PA is given by \cite{pozar2021microwave}:
\begin{align}\label{In_Waveguide_Channel_Model}
h_{\rm{i}}({\bm\psi}_{0}^{\rm{M}},{\bm\psi}^{\rm{M}})\triangleq{{\rm{e}}^{-\frac{2\pi}{\lambda_{\rm{g}}}\lVert{\bm\psi}_{0}^{\rm{M}}-{\bm\psi}^{\rm{M}}\rVert}}
={{\rm{e}}^{-\frac{2\pi}{\lambda_{\rm{g}}}\lvert{\psi}_{0}^{\rm{M}}-{\psi}^{\rm{M}}\rvert}},
\end{align}
where $\lambda_{\rm{g}}\triangleq\frac{\lambda}{n_{\rm{eff}}}$ is the guided wavelength and $n_{\rm{eff}}$ is the effective refractive index of the dielectric waveguide \cite{pozar2021microwave}. The propagation attenuation inside the waveguide is not included in \eqref{In_Waveguide_Channel_Model}. This assumption yields an upper bound on the performance of conventional PASS \cite{ding2024flexible}.
\subsubsection{Signal Model}
Let $s\sim{\mathcal{CN}}(0,1)$ denote the normalized data symbol transmitted by the user. The received signal at the feed point of the monolithic waveguide is given by
\begin{align}\label{Uplink_PASS_Basic_Model}
y_{\rm{M}}=h_{\rm{i}}({\bm\psi}_{0}^{\rm{M}},{\bm\psi}^{\rm{M}})h_{\rm{o}}({\bm\psi}^{\rm{M}},{\mathbf{u}})\sqrt{P}s+n_{\rm{M}},
\end{align}
where $P$ is the transmit power and $n_{\rm{M}}\sim{\mathcal{CN}}(0,\sigma^2)$ represents additive white Gaussian noise (AWGN) with variance $\sigma^2$. The resulting received SNR can be written as follows:
\begin{subequations}\label{Conventional_PASS_SNR_Basic}
\begin{align}
\gamma_{\rm{M}}&=\frac{P}{\sigma^2}\lvert h_{\rm{i}}({\bm\psi}_{0}^{\rm{M}},{\bm\psi}^{\rm{M}})h_{\rm{o}}({\bm\psi}^{\rm{M}},{\mathbf{u}})\rvert^2
=\frac{P\eta}{\sigma^2\lVert{\bm\psi}^{\rm{M}}-{\mathbf{u}}\rVert^2}\\
&=\frac{P\eta}{\sigma^2((\psi^{\rm{M}}-u_x)^2+c_y)},
\end{align}
\end{subequations}
where $c_y\triangleq u_y^2+(d-u_z)^2$. The SNR is maximized at $\psi^{\rm{M}}=u_x$, i.e., when the PA aligns with the projection of the user onto the waveguide, and the maximum SNR is $\gamma_{\rm{M}}=\frac{P\eta}{\sigma^2 c_y}$.
\subsection{SWAN}
\subsubsection{System Architecture}
SWAN addresses the IAR limitation of conventional PASS. SWAN replaces the monolithic waveguide with a segmented structure that consists of multiple short waveguide segments arranged end-to-end. Each segment has a dedicated feed point for signal injection and extraction. SWAN activates at most one PA in each segment, and this design yields a multi-PA uplink model without IAR. The segment outputs are forwarded to the BS through wired links, such as optical fibers or low-loss coaxial cables, as illustrated in {\figurename} {\ref{Figure: New_SWAN_System_Model}}. Since these wired links incur negligible attenuation compared with the dielectric segments, they are modeled as lossless. Furthermore, prior studies indicate that in-waveguide propagation loss within each short segment is negligible under typical SWAN parameters \cite{ouyang2025uplink}, and thus this impact can be safely discarded in the subsequent analyses.

The segmented waveguide consists of $M$ segments, each of length $L$, such that $D_x=LM$. When $M=1$, the segmented structure reduces to a single-waveguide architecture. Let ${\bm\psi}_{0}^{m}\triangleq[\psi_{0}^{m},0,d]^{\mathsf{T}}$ denote the location of the feed point of the $m$th segment, where $\psi_{0}^{1}<\psi_{0}^{2}<\ldots<\psi_{0}^{M}$. The feed point lies at the front-left end of each segment. The location of the PA in the $m$th segment (i.e., the $m$th PA) is denoted by ${\bm\psi}^{m}\triangleq[\psi^{m},0,d]^{\mathsf{T}}$, subject to the following constraints: 
\begin{align}
\psi_{0}^{m}\leq \psi^{m}\leq \psi_{0}^{m}+L,\lvert\psi^{m}-\psi^{m'}\rvert\geq \Delta,\forall m\ne m',
\end{align}
where $\Delta>0$ is the minimum inter-antenna spacing used to suppress mutual coupling \cite{ivrlavc2010toward}.

\begin{figure}[!t]
\centering
    \subfigure[Segment selection.]
    {
        \includegraphics[height=0.11\textwidth]{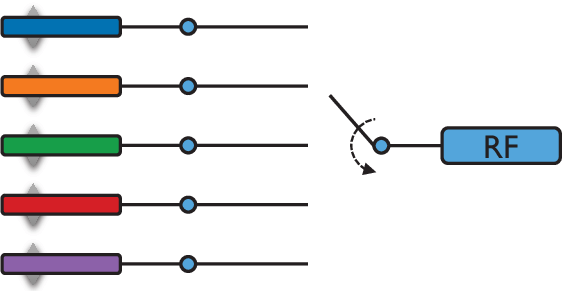}
	   \label{Figure_PAN_Protocol1}
    }
   \subfigure[Segment aggregation.]
    {
        \includegraphics[height=0.11\textwidth]{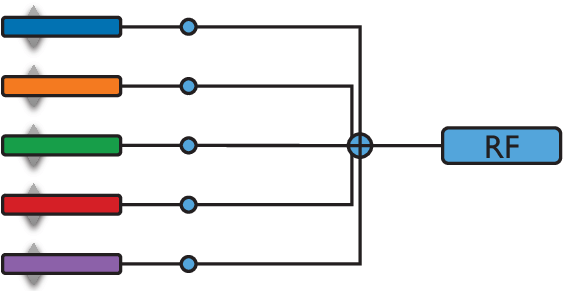}
	   \label{Figure_PAN_Protocol2}
    }
\caption{Illustration of two protocols for operating SWANs.}
\label{Figure: PAN_Protocol}
\vspace{-10pt}
\end{figure}

\subsubsection{Signal Model}\label{System_Model_Uplink_SWAN}
The received signal at the feed point of the $m$th segment is given by
\begin{align}\label{Uplink_PASS_Basic_Model}
y_{m}=h_{\rm{i}}({\bm\psi}_{0}^{m},{\bm\psi}^{m})h_{\rm{o}}({\bm\psi}^{m},{\mathbf{u}})\sqrt{P}s+n_m,
\end{align}
where $P$ denotes the transmit power, and $n_m\sim{\mathcal{CN}}(0,\sigma^2)$ represents AWGN with variance $\sigma^2$. 

Referring to {\figurename} {\ref{Figure: SWAN_System_Model1}}, the performance of SWAN depends on the manner in which the segment feed points connect to the BS radio-frequency (RF) front-end. Two fundamental operating protocols are considered \cite{ouyang2025uplink}: \romannumeral1) \emph{SS (segment selection)} and \romannumeral2) \emph{SA (segment aggregation)}, as illustrated in {\figurename} {\ref{Figure: PAN_Protocol}}.
\begin{itemize}
  \item \emph{Segment Selection:} {\figurename} {\ref{Figure_PAN_Protocol1}} shows that only one segment connects to the RF chain at any time. This protocol relies on a switching network and entails low hardware complexity. Let ${\bar{m}}\in[M]$ denote the index of the selected segment. The received uplink signal is given by
\begin{align}\label{Uplink_PASS_Basic_Model_SS}
y_{\rm{SS}}=y_{{\bar{m}}}=h_{\rm{i}}({\bm\psi}_{0}^{{\bar{m}}},{\bm\psi}^{{\bar{m}}})h_{\rm{o}}({\bm\psi}^{{\bar{m}}},{\mathbf{u}})
\sqrt{P}s+n_{{\bar{m}}}.
\end{align}
  \item \emph{Segment Aggregation:} {\figurename} {\ref{Figure_PAN_Protocol2}} shows that all $M$ feed points connect to a single RF chain through a power splitter. The signals extracted from all segments are aggregated and forwarded for baseband processing. The aggregated received signal is $y_{\rm{SA}}=\sum_{m=1}^{M}y_{m}$. It follows from \eqref{Uplink_PASS_Basic_Model} that
\begin{equation}\label{Uplink_PASS_Basic_Model_SA}
y_{\rm{SA}}=\sum_{m=1}^{M}h_{\rm{i}}({\bm\psi}_{0}^{m},{\bm\psi}^{m})
h_{\rm{o}}({\bm\psi}^{m},{\mathbf{u}})\sqrt{P}s+{n_{\rm{SA}}},
\end{equation}
where $n_{\rm{SA}}\triangleq\sum_{m=1}^{M}n_m\sim{\mathcal{CN}}(0,M\sigma^2)$ denotes the aggregated noise. Since all segments contribute to the received signal, SA can achieve higher performance than SS, at the cost of higher hardware complexity for signal aggregation.
\end{itemize}
\section{Maintainability of Conventional PASS}\label{Section: Maintenance of the Conventional PASS}
After establishing the system model, we analyze the maintainability of the conventional PASS architecture depicted in {\figurename} {\ref{Figure: Conventional_PASS_System_Model}}.
\subsection{Reliability Model of the Waveguide}
\subsubsection{Failure Model}
The waveguide is assumed to experience random failures caused by hardware defects, material degradation, or environmental factors. To capture this stochastic behavior, we model the waveguide lifetime $T_{\rm{M}}^1$ as a random variable. The instantaneous intensity of failure at time $t$, which is the probability per unit time that a failure occurs at time $t$ given that the waveguide remains operational up to $t$, is defined by the following \emph{hazard rate} \cite{trivedi2001probability}:
\begin{align}
\lambda_{\rm{M}}(t)=\lim_{\Delta t\rightarrow0}\frac{\Pr(\left.t\leq T_{\rm{M}}^1\leq t +\Delta t\right\rvert T_{\rm{M}}^1\geq t)}{\Delta t}.
\end{align}
Accordingly, the probability that the waveguide fails within a small time interval $[t,t+\Delta t]$, conditioned on its survival up to time $t$, is given by
\begin{equation}\label{Definition_Failture_Rate}
\begin{split}
&\Pr(\left.{\text{fail~in~}}[t,t+\Delta t]\right\rvert{\text{still~working~at~time~}}t)\\&=
\Pr(\left.t\leq T_{\rm{M}}^1\leq t +\Delta t\right\rvert T_{\rm{M}}^1\geq t)\approx \lambda_{\rm{M}}(t) \Delta t.
\end{split}
\end{equation}
For analytical tractability, $T_{\rm{M}}^1$ is modeled as an exponentially distributed random variable $T_{\rm{M}}^1\sim{\rm{Exp}}(\lambda_{\rm{M}})$ with a constant failure rate $\lambda_{\rm{M}}(t)=\lambda_{\rm{M}}$ (failures per unit time). The cumulative distribution function (CDF) is given by
\begin{align}
\Pr(T_{\rm{M}}^1\leq t)=1-{\rm{e}}^{-\lambda_M t},~t\geq0,
\end{align}
and the probability density function (PDF) is given by
\begin{align}
p_{{\rm{M}}}(t)\triangleq\frac{{\rm{d}}}{{\rm{d}}t}\Pr(T_{\rm{M}}^1\leq t)=\lambda_{\rm{M}}{\rm{e}}^{-\lambda_M t},~t\geq0.
\end{align} 
The \emph{mean time-to-failure (MTTF)} is then obtained as follows:
\begin{align}
{\rm{MTTF}}_{\rm{M}}\triangleq
\int_{0}^{\infty}t\lambda_{\rm{M}}{\rm{e}}^{-\lambda_M t}{\rm{d}}t=\frac{1}{\lambda_{\rm{M}}}.
\end{align}
The exponential lifetime assumption is widely used in reliability engineering because it yields a \emph{memoryless} model, meaning that the residual lifetime of a waveguide is independent of its operational duration. Specifically,
\begin{subequations}
\begin{align}
&\Pr(\left.t\leq T_{\rm{M}}^1\leq t +\Delta t\right\rvert T_{\rm{M}}^1\geq t)\\
&=\frac{\Pr(t\leq T_{\rm{M}}^1\leq t +\Delta t)}{\Pr(T_{\rm{M}}^1\geq t)}\\
&=\frac{\Pr(T_{\rm{M}}^1\leq t +\Delta t)-\Pr(T_{\rm{M}}^1\leq t)}{\Pr(T_{\rm{M}}^1\geq t)}\\
&=\frac{(1-{\rm{e}}^{-\lambda_M (t +\Delta t)})-(1-{\rm{e}}^{-\lambda_M t})}{{\rm{e}}^{-\lambda_M t}}\\
&=1-{\rm{e}}^{-\lambda_M \Delta t}=\Pr(T_{\rm{M}}^1\geq \Delta t).
\end{align}
\end{subequations}
This property is appropriate when failures occur randomly over time without aging effects \cite{trivedi2001probability}. 

The failure rate $\lambda_{\rm{M}}$ is assumed to scale linearly with the physical length of the waveguide. This assumption is consistent with classical reliability models, where the hazard rate increases with the component size that is exposed to potential failures \cite{trivedi2001probability}. Let $\lambda_0$ denote the failure rate per unit length. The failure rate of a monolithic waveguide of length $D_x$ is then given by
\begin{align}\label{Failture_Rate_Conventional_PASS}
\lambda_{\rm{M}}=\lambda_0 D_x.
\end{align}
This linear scaling also follows from the exponential lifetime model. Consider a monolithic waveguide of integer length $D_x$, which can be viewed as a cascade of $D_x$ interconnected unit-length sub-waveguides. Let $T_{k}^{\rm{M}}\sim{\rm{Exp}}(\lambda_0)$ denote the lifetime of the $k$th sub-waveguide for $k\in[D_x]$. The waveguide becomes unavailable once one of the sub-waveguides fails, so the lifetime of the monolithic waveguide is limited by the shortest sub-waveguide lifetime, i.e.,
\begin{align}
T_{\rm{M}}^1=\min\{T_{1}^{\rm{M}},\ldots,T_{D_x}^{\rm{M}}\}.
\end{align}
By the exponential race property \cite{ross1995stochastic}, the minimum of independent exponential random variables remains exponential, and its rate equals the sum of the individual rates. This property yields $T_{\rm{M}}^1\sim{\rm{Exp}}(\sum_{k=1}^{D_x}\lambda_0)={\rm{Exp}}(\lambda_{\rm{M}})$, which confirms that $\lambda_{\rm{M}}=\sum_{k=1}^{D_x}\lambda_0=\lambda_0 D_x$.
\subsubsection{Repair Model}
After a failure, the waveguide is assumed to be repairable. Due to logistical constraints, technician availability, and replacement delays, the repair time $T_{\rm{M}}^{0}$ is modeled as a random variable. Following standard practice in reliability theory \cite{trivedi2001probability}, $T_{\rm{M}}^{0}$ is assumed to be exponentially distributed with rate $\mu_{\rm{M}}$, i.e., $T_{\rm{M}}^{0}\sim{\rm{Exp}}(\mu_{\rm{M}})$. The \emph{mean time-to-repair (MTTR)} is therefore given by
\begin{align}
{\rm{MTTR}}_{\rm{M}}\triangleq
\int_{0}^{\infty}t\mu_{\rm{M}}{\rm{e}}^{-{\mu_M} t}{\rm{d}}t=\frac{1}{\mu_{\rm{M}}}.
\end{align}
This exponential assumption supports tractable analysis and underpins classical Markov reliability models \cite{ross1995stochastic}. 

In practical deployments, repair time accounts for fault detection, isolation, component replacement, and recalibration. These operations typically scale with the physical size of the waveguide. A simple model assumes that the MTTR increases linearly with waveguide length, so that 
\begin{align}
{\rm{MTTR}}_{\rm{M}}=\frac{1}{\mu_{\rm{M}}}=\frac{D_x}{\mu_0}, 
\end{align}
where $\frac{1}{\mu_0}$ denotes the MTTR per unit length, or equivalently, $\mu_0$ represents the repair rate per unit length. Therefore, the repair rate of the entire waveguide is given by
\begin{align}\label{Repair_Rate_Conventional_PASS}
\mu_{\rm{M}}=\frac{1}{{\rm{MTTR}}_{\rm{M}}}=\frac{\mu_0}{D_x}.
\end{align}

\begin{figure}[!t]
\centering
\includegraphics[width=0.3\textwidth]{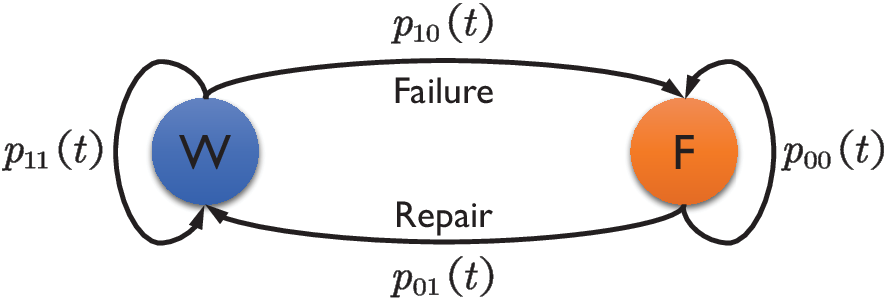}
\caption{Illustration of the CTMC, where ``W'' and ``F'' refer to ``working'' and ``failed'', respectively.}
\label{Figure_Markov}
\vspace{-15pt}
\end{figure}

\subsubsection{Markov Reliability Model}
The operational state of each waveguide can be modeled as a CTMC with two states:
\begin{itemize}
  \item \emph{State 1 (Working):} The waveguide is operational.
  \item \emph{State 0 (Failed):} The waveguide is under failure or repair.
\end{itemize}
The time-dependent transition probabilities of this CTMC are defined as follows:
\begin{subequations}
\begin{align}
&p_{11}(t)\triangleq\Pr(\left.{\text{Working~at~}}t\right|{\text{Working~at~}}0),\\
&p_{10}(t)\triangleq\Pr(\left.{\text{Failed~at~}}t\right|{\text{Working~at~}}0)=1-p_{11}(t),\\
&p_{00}(t)\triangleq\Pr(\left.{\text{Failed~at~}}t\right|{\text{Failed~at~}}0),\\
&p_{01}(t)\triangleq\Pr(\left.{\text{Working~at~}}t\right|{\text{Failed~at~}}0)=1-p_{00}(t).
\end{align}
\end{subequations}
This state transition process is illustrated in {\figurename} {\ref{Figure_Markov}}. 
\vspace{-5pt}
\begin{lemma}\label{Lemma_Waveguide_CTMC_Transition_Probabilities}
For failure rate $\lambda_{\rm{M}}$ and repair rate $\mu_{\rm{M}}$, the CTMC transition probabilities for the monolithic waveguide can be expressed as follows:
\begin{subequations}
\begin{align}
&p_{11}(t)=\frac{\mu_{\rm{M}}}{\lambda_{\rm{M}}+\mu_{\rm{M}}}+\frac{\lambda_{\rm{M}}}{\lambda_{\rm{M}}+\mu_{\rm{M}}}{\rm{e}}^{-(\lambda_{\rm{M}}+\mu_{\rm{M}})t},
\label{Waveguide_CTMC_Transition_Probability1}\\
&p_{10}(t)=\frac{\lambda_{\rm{M}}}{\lambda_{\rm{M}}+\mu_{\rm{M}}}-\frac{\lambda_{\rm{M}}}{\lambda_{\rm{M}}+\mu_{\rm{M}}}{\rm{e}}^{-(\lambda_{\rm{M}}+\mu_{\rm{M}})t},
\label{Waveguide_CTMC_Transition_Probability2}\\
&p_{00}(t)=\frac{\lambda_{\rm{M}}}{\lambda_{\rm{M}}+\mu_{\rm{M}}}+\frac{\mu_{\rm{M}}}{\lambda_{\rm{M}}+\mu_{\rm{M}}}{\rm{e}}^{-(\lambda_{\rm{M}}+\mu_{\rm{M}})t},
\label{Waveguide_CTMC_Transition_Probability3}\\
&p_{01}(t)=\frac{\mu_{\rm{M}}}{\lambda_{\rm{M}}+\mu_{\rm{M}}}-\frac{\mu_{\rm{M}}}{\lambda_{\rm{M}}+\mu_{\rm{M}}}{\rm{e}}^{-(\lambda_{\rm{M}}+\mu_{\rm{M}})t}.
\label{Waveguide_CTMC_Transition_Probability4}
\end{align}
\end{subequations}
\end{lemma}
\vspace{-5pt}
\begin{IEEEproof}
Please refer to Appendix \ref{Proof_Lemma_Waveguide_CTMC_Transition_Probabilities} for more details.
\end{IEEEproof}
The steady-state probabilities of the waveguide being \emph{working} or \emph{failed} are obtained as follows:
\begin{subequations}\label{Staedy_Prob_Basic_Model}
\begin{align}
&P_{\rm{M}}^{1}=\lim_{t\rightarrow\infty}p_{11}(t)=\lim_{t\rightarrow\infty}p_{01}(t)=\frac{\mu_{\rm{M}}}{\lambda_{\rm{M}}+\mu_{\rm{M}}},\\
&P_{\rm{M}}^{0}=\lim_{t\rightarrow\infty}p_{10}(t)=\lim_{t\rightarrow\infty}p_{00}(t)=\frac{\lambda_{\rm{M}}}{\lambda_{\rm{M}}+\mu_{\rm{M}}}.
\end{align}
\end{subequations}
These steady-state values quantify the long-term reliability and downtime ratio of the waveguide. It is worth noting that the working-state probability can also be expressed as $P_{\rm{M}}^{1}=\frac{1/\lambda_{\rm{M}}}{1/\lambda_{\rm{M}}+1/\mu_{\rm{M}}}=\frac{{\rm{MTTF}}_{\rm{M}}}{{\rm{MTTF}}_{\rm{M}}+{\rm{MTTR}}_{\rm{M}}}$.
\vspace{-5pt}
\begin{remark}
The above relationship indicates that the steady-state probability $P_{\rm{M}}^{1}$ represents the fraction of time the system remains operational, which aligns with intuition.
\end{remark}
\vspace{-5pt}
Substituting \eqref{Failture_Rate_Conventional_PASS} and \eqref{Repair_Rate_Conventional_PASS} into \eqref{Staedy_Prob_Basic_Model} yields
\begin{align}\label{Steady_Prob_Refined}
P_{\rm{M}}^{1}=\frac{1}{\frac{\lambda_{\rm{M}}}{\mu_{\rm{M}}}+1}=\frac{1}{\varepsilon_0 D_x^2+1},~P_{\rm{M}}^{0}=1-\frac{1}{\varepsilon_0 D_x^2+1},
\end{align}
where $\varepsilon_0\triangleq\frac{\lambda_0}{\mu_0}$ is referred to as the \emph{failure-repair rate ratio (F3R)}. Equation \eqref{Steady_Prob_Refined} shows that the steady-state probabilities of the waveguide being working or failed depend on the ratio $\varepsilon_0$, rather than on the individual values of the unit-length failure rate $\lambda_0$ and the unit-length repair rate $\mu_0$.
\vspace{-5pt}
\begin{remark}\label{Remark_Steady_State_Probability}
The result in \eqref{Steady_Prob_Refined} admits a clear physical interpretation. The long-term operational state of a waveguide depends jointly on how often failures occur and how quickly repairs can be completed. A small failure rate alone does not guarantee high availability if the repair process is slow, and a high repair rate alone cannot compensate for frequent failures. Based on \eqref{Steady_Prob_Refined}, the steady-state operational probability is determined by the ratio between the failure and repair rates, rather than by their absolute magnitudes. A smaller failure rate $\lambda_0$, together with a larger repair rate $\mu_0$, leads to a higher probability that the waveguide remains in the working state.
\end{remark}
\vspace{-5pt}
Two limiting regimes are of particular interest. When $\varepsilon_0\rightarrow0$, the waveguide operates in a \emph{highly reliable} environment, where failures are rare and repairs are rapid. In this case, $P_{\rm{M}}^{1}\rightarrow1$ and $P_{\rm{M}}^{0}\rightarrow0$, which indicates that the waveguide is \emph{almost always operational and seldom experiences failure}. In contrast, when $\varepsilon_0\rightarrow\infty$, the waveguide operates in a \emph{highly unreliable} environment, where failures occur frequently and repairs are slow. Under this condition, $P_{\rm{M}}^{1}\rightarrow0$ and $P_{\rm{M}}^{0}\rightarrow1$, which implies that the waveguide is \emph{almost always in the failed state and rarely operational}.
\subsection{Performance of Conventional PASS}\label{Section:Covnetional_PASS_Maintance}
After establishing the reliability model, we evaluate the performance of conventional PASS from a maintainability perspective. For analytical tractability, we consider a waveguide that has operated for a sufficiently long time, so that its state follows the steady-state distribution in \eqref{Staedy_Prob_Basic_Model}. 

From \eqref{Conventional_PASS_SNR_Basic}, the maximum received SNR is $\gamma_{\rm{M}}=\frac{P\eta}{\sigma^2 c_y}$. When waveguide availability is taken into account, the instantaneous user rate becomes a binary random variable as follows:
\begin{align}
{\mathcal{R}}_{\rm{M}}\triangleq\left\{\begin{array}{ll}
\log_2(1+\gamma_{\rm{M}})             & {\text{with~probability~}}P_{\rm{M}}^{1}\\
0           & {\text{with~probability~}}P_{\rm{M}}^{0}
\end{array}\right.,
\end{align}
where $P_{\rm{M}}^{1}$ and $P_{\rm{M}}^{0}$ denote the steady-state probabilities that the waveguide is working and failed, respectively. 

Since waveguide unreliability introduces randomness into the achievable rate, we adopt the PNR and the OP to characterize maintainability. The PNR is defined as follows:
\begin{align}\label{PNR_Conventional_PASS}
{\mathcal{P}}_{\rm{M}}^{+}\triangleq\Pr({\mathcal{R}}_{\rm{M}}>0)=P_{\rm{M}}^{1}.
\end{align}
The OP is defined as follows:
\begin{align}\label{OP_Conventional_PASS}
{\mathcal{P}}_{\rm{M}}\triangleq\Pr({\mathcal{R}}_{\rm{M}}<R_0)=\left\{\begin{array}{ll}
P_{\rm{M}}^{0}             & \gamma_{\rm{M}}\geq\tau\\
1           & \gamma_{\rm{M}}<\tau
\end{array}\right.,
\end{align}
where $R_0$ denotes the target data rate and $\tau\triangleq2^{R_0}-1$. \emph{A higher PNR or a lower OP corresponds to higher reliability and stronger maintainability.} For brevity, we consider the case where $\gamma_{\rm{M}}\geq\tau$, i.e., where $\log_2(1+\gamma_{\rm{M}})\geq R_0$. Substituting \eqref{Failture_Rate_Conventional_PASS} and \eqref{Repair_Rate_Conventional_PASS} into \eqref{PNR_Conventional_PASS} and \eqref{OP_Conventional_PASS} yields
\begin{align}
{\mathcal{P}}_{\rm{M}}^{+}&= P_{\rm{M}}^{1}=\frac{\mu_{\rm{M}}}{\lambda_{\rm{M}}+\mu_{\rm{M}}}=\frac{1}{\varepsilon_0 D_x^2+1},\label{PNR_Conventional_PASS_Further}\\
{\mathcal{P}}_{\rm{M}}&= P_{\rm{M}}^{0}=\frac{\lambda_{\rm{M}}}{\lambda_{\rm{M}}+\mu_{\rm{M}}}=1-\frac{1}{\varepsilon_0 D_x^2+1}.\label{OP_Conventional_PASS_Further}
\end{align}
It follows that ${\mathcal{P}}_{\rm{M}}^{+}\simeq{\mathcal{O}}(D_x^{-2})$ and ${\mathcal{P}}_{\rm{M}}\simeq1-{\mathcal{O}}(D_x^{-2})$. As $D_x\rightarrow\infty$, the limiting behavior is given by
\begin{align}\label{OP_Asym_D_x_Conventional_PASS}
\lim\nolimits_{D_x\rightarrow\infty}{\mathcal{P}}_{\rm{M}}^{+}=0,~
\lim\nolimits_{D_x\rightarrow\infty}{\mathcal{P}}_{\rm{M}}=1-0=1. 
\end{align}
This result indicates that, as the waveguide length or service region increases indefinitely, the communication link will \emph{almost certainly experience outage} due to the \emph{limited maintainability} of a long waveguide. Consequently, the effective long-term throughput of conventional PASS asymptotically approaches zero. Moreover, the PNR decays quadratically with the waveguide length, as ${\mathcal{P}}_{\rm{M}}^{+}\simeq{\mathcal{O}}(D_x^{-2})$.
\vspace{-5pt}
\begin{remark}
This degradation follows from the length-dependent failure and repair rates. A longer waveguide yields a higher failure rate and a lower repair rate, which reduces the MTTF and increases the MTTR. Longer structures are more exposed to random failures and require more effort to restore.
\end{remark}
\vspace{-5pt}
\vspace{-5pt}
\begin{remark}
Many existing analyses of conventional PASS emphasize the performance gains obtained from a larger service region or a longer waveguide, and they often omit reliability and maintainability. The above derivation shows that maintainability can dominate long-term performance and should be included in deployment-oriented PASS design.
\end{remark}
\vspace{-5pt}
\section{Maintainability of SWAN}\label{Section:Uplink SWAN}
After analyzing the conventional PASS architecture, we turn to the segmented-waveguide architecture illustrated in {\figurename} {\ref{Figure: New_SWAN_System_Model}}.
\subsection{Reliability Model of the Waveguide Segment}
For analytical tractability, the lifetimes and repair times of all waveguide segments are assumed to be independent and identically distributed (i.i.d.). Let $\lambda_{\rm{S}}$ and $\mu_{\rm{S}}$ denote the failure rate and repair rate of each segment, respectively. Based on the unit-length failure and repair rates $\lambda_0$ and $\mu_0$, the segment-level parameters are expressed as follows:
\begin{align}\label{Rate_Parameter_SWAN}
\lambda_{\rm{S}}=\lambda_0 L= \lambda_0 \frac{D_x}{M},\quad\mu_{\rm{S}}=\frac{\mu_0}{L}=\frac{M\mu_0}{D_x}.
\end{align}
Therefore, the MTTF and MTTR of each segment can be written as follows:
\begin{subequations}\label{Segment_MMTTF_MTTR}
\begin{align}
&{\rm{MTTF}}_{\rm{S}}=\frac{1}{\lambda_{\rm{S}}}= \frac{M}{D_x\lambda_0}=M\times{\rm{MTTF}}_{\rm{M}},\\
&{\rm{MTTR}}_{\rm{S}}=\frac{1}{\mu_{\rm{S}}}=\frac{D_x}{M\mu_0}=\frac{{\rm{MTTR}}_{\rm{M}}}{M}.
\end{align}
\end{subequations}
Since $M\geq1$, it follows that ${\rm{MTTF}}_{\rm{S}}\geq{\rm{MTTF}}_{\rm{M}}$ and ${\rm{MTTR}}_{\rm{S}}\leq{\rm{MTTR}}_{\rm{M}}$. This behavior is consistent with physical intuition. Short waveguide segments are less susceptible to random failures and can be repaired more quickly than a single long, monolithic waveguide. Let $P_{m}^1\triangleq\frac{\mu_{\rm{S}}}{\lambda_{\rm{S}}+\mu_{\rm{S}}}=\frac{1}{\varepsilon_0 L^2+1}$ and $P_{m}^0\triangleq\frac{\lambda_{\rm{S}}}{\lambda_{\rm{S}}+\mu_{\rm{S}}}=\frac{\varepsilon_0 L^2}{\varepsilon_0 L^2+1}$ denote the steady-state probabilities that the $m$th segment is working and failed, respectively, where $m\in[M]$. These expressions show that the steady-state reliability of each segment is also governed by the F3R $\varepsilon_0$, which is consistent with the interpretation in Remark \ref{Remark_Steady_State_Probability}.

In the following, we analyze the maintainability of the segmented structure under the SS and SA protocols.
\subsection{Segment Selection}\label{Section: Uplink SWAN: Segment Switching}
\subsubsection{Optimal Antenna Activation}
For the SS protocol, only one segment and its associated PA are activated. From \eqref{Uplink_PASS_Basic_Model_SS}, the maximum received SNR can be expressed as follows:
\begin{subequations}\label{SNR_SS_Uplink_Definition}
\begin{align}
\gamma_{\rm{SS}}&\triangleq\max_{\psi^{m}\in[\psi_{0}^{m},\psi_{0}^{m}+L]}
\frac{P\lvert h_{\rm{i}}({\bm\psi}_{0}^{{{m}}},{\bm\psi}^{{{m}}})h_{\rm{o}}({\bm\psi}^{{{m}}},{\mathbf{u}})\rvert^2}{\sigma^2}\\
&=\max_{\psi^{m}\in[\psi_{0}^{m},\psi_{0}^{m}+L]}
\frac{P\eta}{\sigma^2((u_x-{\psi}^{m})^2+c_y)}.
\end{align}
\end{subequations}
The optimal PA should be activated within the segment closest to the user, whose index is determined by
\begin{align}\label{SS_Optimal_Uplink}
m^{\star}=\left\lceil\frac{u_x-{\psi}_{0}^{1}}{L}\right\rceil.
\end{align}  
Substituting this into \eqref{SNR_SS_Uplink_Definition} gives
\begin{align}
\gamma_{\rm{SS}}=\max_{{\psi}^{m^{\star}}\in[\psi_{0}^{m^{\star}},\psi_{0}^{m^{\star}}+L]}
\frac{P\eta}{\sigma^2((u_x-{\psi}^{m^{\star}})^2+c_y)}.
\end{align}
The optimal PA location satisfies ${\psi}^{m^{\star}}=u_x$, and the maximum SNR is given by $\gamma_{\rm{SS}}=\frac{P\eta}{\sigma^2 c_y}$. Notice that $\gamma_{\rm{SS}}=\gamma_{\rm{M}}$.
\subsubsection{Performance of SWAN Under SS}
By incorporating the reliability model, the instantaneous user rate becomes a binary random variable as follows:
\begin{align}
{\mathcal{R}}_{\rm{SS}}\triangleq\left\{\begin{array}{ll}
\log_2(1+\gamma_{\rm{SS}})             & {\text{with~probability~}}P_{m^{\star}}^{1}\\
0           & {\text{with~probability~}}P_{m^{\star}}^{0}
\end{array}\right..
\end{align}
The PNR is therefore expressed as follows:
\begin{align}\label{SS_PNR_Expression}
{\mathcal{P}}_{\rm{SS}}^{+}\triangleq \Pr({\mathcal{R}}_{\rm{SS}}>0)=P_{m^{\star}}^{1}=1-\frac{1}{1+\frac{1}{L^2\varepsilon_0}},
\end{align}
where $L=\frac{D_x}{M}$. The corresponding OP is given by
\begin{align}\label{SS_OP_Expression}
{\mathcal{P}}_{\rm{SS}}\triangleq \Pr({\mathcal{R}}_{\rm{SS}}<R_0)=P_{m^{\star}}^{0}=\frac{1}{1+\frac{1}{L^2\varepsilon_{0}}},
\end{align}
where $\gamma_{\rm{SS}}\geq2^{R_0}-1$, i.e., $\log_2(1+\gamma_{\rm{SS}})\geq R_0$. 

Since ${\mathcal{P}}_{\rm{SS}}=1-{\mathcal{P}}_{\rm{SS}}^{+}$, it suffices to study ${\mathcal{P}}_{\rm{SS}}^{+}$. Inserting $L=\frac{D_x}{M}$ into \eqref{SS_PNR_Expression} gives ${\mathcal{P}}_{\rm{SS}}^{+}=\frac{M^2}{M^2+D_x^2\varepsilon_0}$. Differentiating ${\mathcal{P}}_{\rm{SS}}^{+}$ versus $M$ yields
\begin{align}\label{SS_OP_Deriviate}
\frac{\partial {\mathcal{P}}_{\rm{SS}}^{+}}{\partial M}=\frac{2\varepsilon_0D_x^2M}{(M^2+\varepsilon_0D_x^2)^2}>0.
\end{align}
Therefore, ${\mathcal{P}}_{\rm{SS}}^{+}$ increases with the number of segments $M$. Equivalently, the OP decreases as $M$ increases.
\vspace{-5pt}
\begin{remark}
The above results confirm that partitioning a long waveguide into a larger number of shorter segments improves the PNR of the PASS communication link. This is intuitive, as more segments reduce the per-segment failure rate and accelerate repair processes under SS, which enhances both the maintainability efficiency and the overall operational reliability of the system, as also indicated by \eqref{Segment_MMTTF_MTTR}.
\end{remark}
\vspace{-5pt}
Furthermore, taking the limit as $M\rightarrow\infty$ yields
\begin{subequations}\label{SS_OP_Asym}
\begin{align}
\lim\nolimits_{M\rightarrow\infty}{\mathcal{P}}_{\rm{SS}}^{+}&=\lim\nolimits_{M\rightarrow\infty}P_{m^{\star}}^{1}\simeq1-{\mathcal{O}}(M^{-2})\rightarrow1,\label{SS_PNR_Asym}\\
\lim\nolimits_{M\rightarrow\infty}{\mathcal{P}}_{\rm{SS}}&=\lim\nolimits_{M\rightarrow\infty}P_{m^{\star}}^{0}\simeq{\mathcal{O}}(M^{-2})\rightarrow0.\label{SS_OP_Asym}
\end{align}
\end{subequations}
This asymptotic behavior is physically meaningful. As the number of segments increases, each segment becomes shorter, which reduces the segment failure rate and increases the segment repair rate under the adopted length-scaling model. Therefore, the system remains operational for almost all time, with a working probability approaching one. In this regime, the impact of maintainability becomes negligible, and the achievable rate approaches the reliability-free limit $\log_2(1+\gamma_{\rm{SS}})$.

The preceding analysis considered a fixed service-region side length $D_x$. We now consider the case where the segment length $L=\frac{D_x}{M}$ is fixed. Under this setting, the PNR of SWAN under SS is independent of the service-region width $D_x$; see \eqref{SS_PNR_Expression}. This result follows from the fact that uplink reception under SS relies only on the selected segment, and its reliability is determined by the segment length $L$ rather than by the total waveguide length. In contrast, Section \ref{Section:Covnetional_PASS_Maintance} shows that the PNR of conventional PASS decreases monotonically with $D_x$ and converges to zero as $D_x\rightarrow\infty$; see \eqref{OP_Asym_D_x_Conventional_PASS}. This degradation occurs because conventional PASS relies on a single continuous waveguide whose failure rate increases with its physical length, while its repair rate decreases under the adopted length-scaling model.
\vspace{-5pt}
\begin{remark}
Under SS, SWAN performance depends only on the fixed segment length $L$, which is much smaller than the monolithic waveguide length $D_x$ in conventional PASS. The segmented structure reduces both the likelihood of failure and the cost associated with repair or replacement. SWAN therefore offers stronger maintainability robustness and sustains stable performance as the service region expands.
\end{remark}
\vspace{-5pt}
\subsubsection{Comparison with Conventional PASS}\label{Section: SS: Comparison with the Conventional PASS}
Since ${\mathcal{P}}_{\rm{M}}^{+}={\mathcal{P}}_{\rm{SS}}^{+}$ when $M=1$, it follows directly that ${\mathcal{P}}_{\rm{SS}}^{+}>{\mathcal{P}}_{\rm{M}}^{+}$ for all $M>1$. The relative PNR gain of SS-based SWAN over conventional PASS is defined as follows:
\begin{align}\label{SS_PNR_Performance_Gain}
{\mathcal{A}}_{\rm{SS}}\triangleq\frac{{\mathcal{P}}_{\rm{SS}}^{+}}{{\mathcal{P}}_{\rm{M}}^{+}}=
\frac{\varepsilon_0 {D_x^2}M^2+M^2}{\varepsilon_0 {D_x^2}+{M^2}}=
\frac{\varepsilon_0 {L^2M^2}+1}{\varepsilon_0 {L^2}+{1}}.
\end{align}
We first consider the case where the side length $D_x$ is fixed. Under this condition, the performance gain ${\mathcal{A}}_{\rm{SS}}$ increases with the number of segments $M$, since $\frac{\partial {\mathcal{A}}_{\rm{SS}}}{\partial M}=\frac{2M(1+\varepsilon_0D_x^2)\varepsilon_0D_x^2}{(\varepsilon_0D_x^2+M^2)^2}>0$. Moreover, $\lim_{M\rightarrow\infty}{\mathcal{A}}_{\rm{SS}}=\varepsilon_0 {D_x^2}+1$, which scales quadratically with $D_x$. Next, for a fixed $M$, differentiating ${\mathcal{A}}_{\rm{SS}}$ w.r.t. $D_x$ yields $\frac{\partial {\mathcal{A}}_{\rm{SS}}}{\partial D_x}
=\frac{2\varepsilon_0D_xM^2(M^2-1)}{(M^2+\varepsilon_0D_x^2)^2}>0$. 
\vspace{-5pt}
\begin{remark}\label{Remark_Performance_Gain_Side_Length}
The above arguments imply that the relative maintainability gain of SS-based SWAN over conventional PASS increases with the number of segments and with the size of the service region.
\end{remark}
\vspace{-5pt}
We next consider the case where the segment length $L$ is fixed. It follows from \eqref{SS_PNR_Performance_Gain} that ${\mathcal{A}}_{\rm{SS}}$ increases with the number of segments $M$. Moreover, ${\mathcal{A}}_{\rm{SS}}\simeq{\mathcal{O}}(M^2)$ as $M\rightarrow\infty$. This behavior indicates that the maintainability advantage of SS-based SWAN over conventional PASS becomes more pronounced as the number of segments grows. A larger $M$ yields a higher PNR under the adopted reliability model.

In the sequel, we analyze two limiting regimes to provide additional insight. Consider $\varepsilon_0\rightarrow0$, which corresponds to a highly reliable environment where failures are rare and repairs are fast. In this case, $\lim_{\varepsilon_0\rightarrow0}{\mathcal{A}}_{\rm{SS}}=\frac{1+0}{1+0}=1$, so the maintainability gain of SWAN over conventional PASS vanishes. This result matches intuition, since segmentation offers little benefit when reliability constraints are weak. 

Next, consider $\varepsilon_0\rightarrow\infty$, which corresponds to a highly unreliable environment where failures are frequent and repairs are slow. Under this condition, both architectures suffer low availability, and the PNRs satisfy $\lim_{\varepsilon_0\rightarrow\infty}{\mathcal{P}}_{\rm{M}}^{+}=\lim_{\varepsilon_0\rightarrow\infty}{\mathcal{P}}_{\rm{SS}}^{+}=0$. However, the relative PNR gain satisfies
\begin{align}\label{SS_Asym_Gain_Environment}
\lim_{\varepsilon_0\rightarrow\infty}{\mathcal{A}}_{\rm{SS}}=\lim_{\varepsilon_0\rightarrow\infty}\frac{\varepsilon_0 {L^2M^2}+1}{\varepsilon_0 {L^2}+{1}}=M^2.
\end{align}
Therefore, even in a highly unreliable environment, segmentation yields an $M^2$-fold improvement in PNR relative to the monolithic-waveguide design. This improvement follows from the reduced failure exposure and the faster recovery associated with shorter, independently maintained segments.
\subsection{Segment Aggregation}
We now analyze the SA-based SWAN architecture. From \eqref{Uplink_PASS_Basic_Model_SA}, the maximum received SNR under SA is given by
\begin{subequations}\label{SNR_Uplink_SA_SWAN_Standard}
\begin{align}
\gamma_{\rm{SA}}&\triangleq\max_{{\bm\psi}\in{\mathcal{X}}_{\rm{SA}}^{\rm{ul}}}
\frac{P\left\lvert \sum_{m=1}^{M}h_{\rm{i}}({\bm\psi}_{0}^{m},{\bm\psi}^{m})
h_{\rm{o}}({\bm\psi}^{m},{\mathbf{u}})\right\rvert^2}{M\sigma^2}\\
&=\max_{{\bm\psi}\in{\mathcal{X}}_{\rm{SA}}^{\rm{ul}}}
\frac{P\left\lvert \sum\limits_{m=1}^{M}\frac{\eta^{\frac{1}{2}}{\rm{e}}^{-{\rm{j}}k_0(\sqrt{(u_x-{\psi}^{m})^2+c_y}+n_{\rm{eff}}({\psi}^{m}-{\psi}_0^{m})) }}{\sqrt{(u_x-{\psi}^{m})^2+c_y}}\right\rvert^2}{M\sigma^2},
\end{align}
\end{subequations}
where ${\bm\psi}\triangleq[{\psi}^{1},\ldots,{\psi}^{M}]^{\mathsf{T}}\in{\mathbbmss{R}}^{M\times1}$ and
\begin{align}
{\mathcal{X}}_{\rm{SA}}^{\rm{ul}}\triangleq\left\{{\bm\psi}\left\lvert\begin{matrix}{\psi}^{m}\in[\psi_{0}^{m},\psi_{0}^{m}+L],m\in[M],\\
\lvert{\psi}^{m}-{\psi}^{m'}\rvert\geq\Delta,m\ne m'\end{matrix}\right.\right\}.
\end{align}
\subsubsection{Optimal Antenna Placement}\label{SA_Antenna_Optimization_Uplink}
To maximize the received SNR under SA, the PA positions must be optimized to achieve constructive signal superposition while minimizing free-space path loss \cite{xu2024rate}. For simplicity, let $M$ be an odd integer.

We first identify the segment containing the projection of the user onto the waveguide. The corresponding index is given by $m^{\star}=\left\lceil\frac{u_x-{\psi}_{0}^{1}}{L}\right\rceil$; see \eqref{SS_Optimal_Uplink}. The PA in the $m^{\star}$th segment is placed directly beneath the user projection, i.e., ${\psi}^{m^{\star}}=u_x$. Next, consider the $(m^{\star}-1)$th segment. To reduce path loss while satisfying the minimum spacing constraint $\Delta$, the PA position is initialized as follows:
\begin{align}\label{SA_Optimization_Update_First}
{\psi}^{m^{\star}-1}=\min\{{\psi}_{0}^{m^{\star}-1}+L,{\psi}^{m^{\star}}-\Delta\}\triangleq {\hat{\psi}}^{m^{\star}-1}.
\end{align} 
This position is then refined to ensure phase alignment between the signals received by the $m^{\star}$th and $(m^{\star}-1)$th PAs. Specifically, the PA at ${\hat{\psi}}^{m^{\star}-1}$ is shifted leftward by a distance $\nu^{m^{\star}-1}>0$ such that
\begin{equation}\label{SA_Slight_Adjustment}
\begin{split}
&({({\hat{\psi}}^{m^{\star}-1}-\nu^{m^{\star}-1}-u_x)^2+c_y})^{1/2}+n_{\rm{eff}}({\hat{\psi}}^{m^{\star}-1}-\nu^{m^{\star}-1}\\
&-{\psi}_{0}^{m^{\star}-1})
=d^{m^{\star}-1}+((d^{m^{\star}}-d^{m^{\star}-1}))\bmod \lambda)\triangleq \hat{d}^{m^{\star}-1},
\end{split}
\end{equation}
where $d^{m^{\star}-1}\triangleq({(\hat{\psi}^{m^{\star}-1}-u_x)^2+c_y})^{1/2}+n_{\rm{eff}}(\hat{\psi}^{m^{\star}-1}-{\psi}_{0}^{m^{\star}-1})$ and $d^{m^{\star}}\triangleq\sqrt{c_y}+n_{\rm{eff}}(u_x-{\psi}_{0}^{m^{\star}})$. The resulting closed-form solution for the phase-alignment shift $\nu^{m^{\star}-1}$ is given by
\begin{equation}\label{SA_Slight_Adjustment_Solution}
\begin{split}
&\nu^{m^{\star}-1}=\\&\left\{\begin{array}{ll}
\hat{\psi}^{m^{\star}-1}-\frac{{\psi}_{0}^{m^{\star}-1}n_{\rm{eff}}^2+\hat{d}^{m^{\star}-1}n_{\rm{eff}}-u_x-\sqrt{\Delta^{m^{\star}-1}}}{n_{\rm{eff}}^2-1}             & {n_{\rm{eff}}\ne1}\\
\hat{\psi}^{m^{\star}-1}-\frac{({\psi}_{0}^{m^{\star}-1}+\hat{d}^{m^{\star}-1})^2-(u_x^2+c_y)}{2({\psi}_{0}^{m^{\star}-1}+\hat{d}^{m^{\star}-1}-u_x)}           & {n_{\rm{eff}}=1}
\end{array}\right.,
\end{split}
\end{equation}
where $\Delta^{m^{\star}-1}\triangleq n_{\rm{eff}}^2(u_x-{\psi}^{m^{\star}-1})^2-2\hat{d}^{m^{\star}-1}n_{\rm{eff}}(u_x-{\psi}_{0}^{m^{\star}-1})+c_y(n_{\rm{eff}}^2-1)+(\hat{d}_{1}^{m^{\star}-1})^2$. Since a propagation distance of one wavelength corresponds to a $2\pi$ phase shift, the optimal shift $\nu^{m^{\star}-1}$ lies on the wavelength scale. This scale is much smaller than the waveguide deployment height $d$, which implies that the resulting change in free-space path loss is negligible. The PA position is then updated as ${\psi}^{m^{\star}-1}=\hat{\psi}^{m^{\star}-1}-\nu^{m^{\star}-1}$. The same procedure applies to the $(m^{\star}-2)$th segment, where the initial position is given by
\begin{align}
{\psi}^{m^{\star}-2}=\min\{{\psi}_{0}^{m^{\star}-2}+L,{\psi}^{m^{\star}-1}-\Delta\}\triangleq\hat{\psi}^{m^{\star}-2}, 
\end{align}
and the corresponding phase-alignment shift $\nu^{m^{\star}-2}$ is computed in an analogous manner using \eqref{SA_Slight_Adjustment} and \eqref{SA_Slight_Adjustment_Solution}. This iterative design continues for all segments with indices $m\leq m^{\star}$ and extends symmetrically to the segments with $m> m^{\star}$.
\subsubsection{Performance of SWAN Under SA}\label{Section: Uplink SA-Based SWAN: Discussion on the Maximum SNR}
The optimization procedure described above enables constructive combination of the received signals \cite{xu2024rate}. As a result, the dual phase shifts induced by signal propagation both inside and outside the waveguide, i.e., $k_0\sqrt{(u_x-{\psi}^{m})^2+c_y}+k_0n_{\rm{eff}}({\psi}^{m}-{\psi}_{0}^{m})$, have no impact on the received SNR. By \eqref{SNR_Uplink_SA_SWAN_Standard}, we express the resulting SNR as follows: 
\begin{align}
\gamma_{\rm{SA}}=
\frac{P\eta}{M\sigma^2}\left( \sum\limits_{m=1}^{M}\frac{1}{({(\hat{\psi}^{m}-\nu^{m}-u_x)^2+c_y})^{1/2}}\right)^2,
\end{align} 
where $\hat{\psi}^{m^{\star}}={\psi}^{m^{\star}}$ and $\nu^{m^{\star}}=0$. As noted earlier, the shift $\nu^{m}$ lies on the wavelength scale, and its impact on free-space path loss is negligible. Therefore, the SNR admits the following approximation:
\begin{align}\label{SNR_SA_Uplink_Approximation_First}
\gamma_{\rm{SA}}\approx
\frac{P\eta}{M\sigma^2}\left( \sum\limits_{m=1}^{M}\frac{1}{({(\hat{\psi}^{m}-u_x)^2+c_y})^{1/2}}\right)^2.
\end{align}
By further taking into account the impact of the segment reliability, we rewrite \eqref{SNR_SA_Uplink_Approximation_First} as the following random variable:
\begin{align}\label{SNR_SA_Uplink_Approximation_First_Random}
\gamma_{\rm{SA}}\approx
\frac{P\eta}{\hat{M}\sigma^2}\left( \sum_{m=1}^{M}\frac{b_m}{({(\hat{\psi}^{m}-u_x)^2+c_y})^{1/2}}\right)^2,
\end{align}
where $\{b_m\}_{m=1}^{M}$ is the set of $M$ i.i.d. binary random variables with $\Pr(b_m=0)=P_{m}^0=\frac{\lambda_{\rm{S}}}{\lambda_{\rm{S}}+\mu_{\rm{S}}}$ and $\Pr(b_m=1)=P_{m}^1=\frac{\mu_{\rm{S}}}{\lambda_{\rm{S}}+\mu_{\rm{S}}}$, and $\hat{M}\triangleq\sum_{m=1}^{M}b_m$. It is worth noting that only the segments in the working state contribute to the received signal. As a result, the effective noise power scales with the number of active segments and changes from $M\sigma^2$ to $\hat{M}\sigma^2$. The PNR under the SA protocol is defined as follows:
\begin{align}
{\mathcal{P}}_{\rm{SA}}^{+}\triangleq \Pr(\log_2(1+\gamma_{\rm{SA}})>0)=\Pr(\gamma_{\rm{SA}}>0),
\end{align}
and the corresponding OP is given by
\begin{align}
{\mathcal{P}}_{\rm{SA}}\triangleq \Pr(\log_2(1+\gamma_{\rm{SA}})<R_0)=\Pr(\gamma_{\rm{SA}}<\tau),
\end{align}
where $R_0$ denotes the target data rate and $\tau=2^{R_0}-1$.
\vspace{-5pt}
\begin{theorem}
The probability of achieving a non-zero rate under the SA protocol is given by
\begin{align}\label{PNR_SA_Basic}
{\mathcal{P}}_{\rm{SA}}^{+}=1-\left(\frac{\lambda_{\rm{S}}}{\lambda_{\rm{S}}+\mu_{\rm{S}}}\right)^M=1-\left(\frac{\varepsilon_0 D_x^2}{\varepsilon_0 D_x^2+M^2}\right)^M.
\end{align}
\end{theorem}
\vspace{-5pt}
\begin{IEEEproof}
A non-zero rate is achieved if and only if at least one segment is in the woking state, which corresponds to the event that there exists an index $m\in[M]$ such that $b_m=1$. Therefore, ${\mathcal{P}}_{\rm{SA}}^{+}=\Pr(\exists m\in[M]:b_m=1)=1-\Pr(b_1=\ldots=b_M=0)$. Since the random variables $\{b_m\}_{m=1}^{M}$ are statistically independent, it follows that $\Pr(b_1=\ldots=b_M=0)=\prod_{m=1}^{M}\Pr(b_m=0)=(\frac{\lambda_{\rm{S}}}{\lambda_{\rm{S}}+\mu_{\rm{S}}})^M$.
\end{IEEEproof}
Given the service region width $D_x$, ${\mathcal{P}}_{\rm{SA}}^{+}$ increases with the number of segments $M$, since $\frac{\partial {\mathcal{P}}_{\rm{SA}}^{+}}{\partial M}=\left(\frac{\varepsilon_0 D_x^2}{\varepsilon_0 D_x^2+M^2}\right)^M\left(\frac{2M^2}{\varepsilon_0 D_x^2+M^2}+\ln\frac{\varepsilon_0 D_x^2+M^2}{\varepsilon_0 D_x^2}\right)>0$. This result indicates that partitioning a long waveguide into a larger number of segments improves the overall system reliability. Moreover, since ${\mathcal{P}}_{\rm{SA}}^{+}={\mathcal{P}}_{\rm{SS}}^{+}$ when $M=1$, it follows that ${\mathcal{P}}_{\rm{SA}}^{+}>{\mathcal{P}}_{\rm{SS}}^{+}$ for $M>1$. This observation confirms that SA outperforms SS in terms of reliability, which is expected because SA exploits all working segments to receive uplink signals. From \eqref{PNR_SA_Basic}, we further obtain
\begin{align}\label{SA_PNR_Asym}
\lim\nolimits_{M\rightarrow\infty}{\mathcal{P}}_{\rm{SA}}^{+}\simeq1-{\mathcal{O}}((\varepsilon_0 D_x^2M^{-2})^{M})\rightarrow1.
\end{align}
This convergence is much faster than that of ${\mathcal{P}}_{\rm{SS}}^{+}$ in \eqref{SS_PNR_Asym}. This result further confirms the effectiveness of activating multiple segments in SWAN to enhance system maintainability.

We next analyze the OP. It can be expressed as follows:
\begin{align}
{\mathcal{P}}_{\rm{SA}}=\Pr\left( \Lambda_M<{\sqrt{\hat{M}\tau}\sigma}/{\sqrt{P\eta}}\right),
\end{align}
where $\Lambda_M\triangleq\sum_{m=1}^{M}\frac{b_m}{({(\hat{\psi}^{m}-u_x)^2+c_y})^{1/2}}$. Using the statistical independence of $\{b_m\}_{m=1}^{M}$, the OP can be written as follows:
\begin{equation}
\begin{split}
{\mathcal{P}}_{\rm{SA}}&=\sum\nolimits_{{\mathcal{I}}\subseteq[M]}\left(P_{\rm{S}}^0\right)^{M-\lvert{\mathcal{I}}\rvert}\left(P_{\rm{S}}^1\right)^{\lvert{\mathcal{I}}\rvert}\\
&\times{\mathbbmss{1}}_{\left\{
\sum_{m\in{\mathcal{I}}}\frac{1}{({(\hat{\psi}^{m}-u_x)^2+c_y})^{1/2}}<\frac{\sqrt{\lvert{\mathcal{I}}\rvert\tau}\sigma}{\sqrt{P\eta}}\right\}},
\end{split}
\end{equation}
where ${\mathbbmss{1}}_{\{\cdot\}}$ denotes the indicator function, $P_{\rm{S}}^1\triangleq\frac{\mu_{\rm{S}}}{\lambda_{\rm{S}}+\mu_{\rm{S}}}$, and $P_{\rm{S}}^0\triangleq\frac{\lambda_{\rm{S}}}{\lambda_{\rm{S}}+\mu_{\rm{S}}}$. The above expression incurs exponential computational complexity and does not yield clear analytical insight. To obtain a tractable characterization, we derive an upper bound on ${\mathcal{P}}_{\rm{SA}}$. Since $M_{\Lambda}\leq M$, the SNR in \eqref{SNR_SA_Uplink_Approximation_First_Random} satisfies $\gamma_{\rm{SA}}\approx\frac{P\eta}{\hat{M}\sigma^2}\Lambda_M^2\leq\frac{P\eta}{M\sigma^2}\Lambda_M^2$. Accordingly, the OP admits the following upper bound:
\begin{align}
{\mathcal{P}}_{\rm{SA}}\leq\Pr\left( \Lambda_M<{\sqrt{M\tau}\sigma}/{\sqrt{P\eta}}\right)\triangleq\hat{\mathcal{P}}_{\rm{SA}}.
\end{align}
Although $\hat{\mathcal{P}}_{\rm{SA}}$ is simpler than the exact OP, it remains analytically intractable. Therefore, we employ a Gaussian approximation based on moment matching to obtain a closed-form outage approximation. 

The mean and variance of $\Lambda_M$ are given by
\begin{subequations}
\begin{align}
{\mathbbmss{E}}\{\Lambda_M\}&=P_{\rm{S}}^1\sum_{m=1}^{M}\frac{1}{({(\hat{\psi}^{m}-u_x)^2+c_y})^{1/2}}\triangleq\mu_{\Lambda_M},\\
{\mathbbmss{V}}\{\Lambda_M\}&=P_{\rm{S}}^0P_{\rm{S}}^1\sum_{m=1}^{M}\frac{1}{{(\hat{\psi}^{m}-u_x)^2+c_y}}\triangleq\sigma_{\Lambda_M}^2,
\end{align}
\end{subequations}
respectively. We approximate $\Lambda_M$ by a Gaussian random variable with mean $\mu_{\Lambda_M}$ and variance $\sigma_{\Lambda_M}^2$. The OP's upper bound can then be approximated as follows:
\begin{align}\label{OP_Upper_Bound_Gaussian_Approximation}
\hat{\mathcal{P}}_{\rm{SA}}\approx\Phi\left( \frac{{\sqrt{M\tau}\sigma}/{\sqrt{P\eta}}-\mu_{\Lambda_M}}{{\sigma_{\Lambda_M}}}\right),
\end{align}
where $\Phi(\cdot)$ denotes the standard Gaussian CDF.

To obtain further insight, we consider a symmetric user placement scenario in which the user is located directly beneath the center of the $\left(\frac{M+1}{2}\right)$th waveguide segment. This yields $m^{\star}=\frac{M+1}{2}$ and $u_x=\psi_0^{m^{\star}}+\frac{L}{2}$. We also assume $L\gg\Delta$, which represents a mild and practical condition. Under these assumptions, the antenna locations satisfy ${\psi}^{m^{\star}+\hat{m}}=\psi_0^{m^{\star}+\hat{m}}$ and $\psi^{m^{\star}-\hat{m}}=\psi_0^{m^{\star}-\hat{m}}+L$ for $\hat{m}=1,\ldots,\frac{M-1}{2}$. Substituting these relations into \eqref{SNR_SA_Uplink_Approximation_First_Random} yields
\begin{align}
\gamma_{\rm{SA}}\approx
\frac{P\eta}{\hat{M}\sigma^2}\left(\frac{b_{m^{\star}}}{\sqrt{c_y}}+\sum\limits_{\hat{m}=1}^{\frac{M-1}{2}}\frac{b_{m^{\star}-\hat{m}}+b_{m^{\star}+\hat{m}}}{({(L(\hat{m}-\frac{1}{2}))^2+c_y})^{\frac{1}{2}}}\right)^2.
\end{align}
Under this symmetric configuration, the mean and variance simplify to the following \cite[Appendix C]{ouyang2025uplink}:
\begin{subequations}\label{Mean_Variance_SA_Channel_Gain}
\begin{align}
\mu_{\Lambda_M}&\approx P_{\rm{S}}^1\left(\frac{1}{\sqrt{c_y}}+\frac{2}{L}\sinh^{-1}\left(\frac{L(M-1)}{2\sqrt{c_y}}\right)\right),\\
\sigma_{\Lambda_M}^2&\approx\frac{P_{\rm{S}}^0P_{\rm{S}}^1}{\sqrt{c_y}}\left(\frac{1}{{\sqrt{c_y}}}+\frac{2}{L}\tan^{-1}\left(\frac{L(M-1)}{2\sqrt{c_y}}\right)\right).
\end{align}
\end{subequations}

By fixing $D_x$ and letting $M\rightarrow\infty$, we obtain
\begin{align}
\frac{\mu_{\Lambda_M}}{\sqrt{M}}\simeq\tau_{D_x}\sqrt{M}\rightarrow\infty,\quad\frac{\sigma_{\Lambda_M}^2}{M}\simeq\frac{\kappa_{D_x}}{2M^2}\rightarrow0,
\end{align}
where $\tau_{D_x}\triangleq\frac{2}{D_x}\sinh^{-1}\left(\frac{D_x}{2\sqrt{c_y}}\right)>0$ and $\kappa_{D_x}\triangleq\frac{4\varepsilon_0D_x}{\sqrt{c_y}}\tan^{-1}\left(\frac{D_x}{2\sqrt{c_y}}\right)>0$. These asymptotic behaviors imply 
\begin{align}
\lim_{M\rightarrow\infty}\frac{\frac{\sqrt{M\tau}\sigma}{\sqrt{P\eta}}-\mu_{\Lambda_M}}{{\sigma_{\Lambda_M}}}=
\lim_{M\rightarrow\infty}\frac{\frac{\sqrt{\tau}\sigma}{\sqrt{P\eta}}-\frac{\mu_{\Lambda_M}}{\sqrt{M}}}{{\sigma_{\Lambda_M}}/\sqrt{M}}\rightarrow-\infty.
\end{align}
Together with the Gaussian tail expansion $\lim_{x\rightarrow-\infty}\Phi(x)\simeq\frac{-1}{\sqrt{2\pi}x{\rm{e}}^{{x^2}/{2}}}$ \cite[Section 7.6]{olver2010nist}, the OP's upper bound satisfies
\begin{align}\label{SA_OP_Asym}
\lim_{M\rightarrow\infty}\hat{\mathcal{P}}_{\rm{SA}}\simeq{\mathcal{O}}(M^{-\frac{3}{2}}
{\rm{e}}^{-{M^2}{\kappa_{D_x}^{-1}}(\sqrt{M}\tau_{D_x}-\frac{\sqrt{\tau}\sigma}{\sqrt{P\eta}})^2})\rightarrow0.
\end{align}
This analysis shows that the OP upper bound for SA converges to zero as the number of segments increases, and it therefore implies that the OP itself also vanishes asymptotically. 
\vspace{-5pt}
\begin{remark}
The above results highlight the effectiveness of segmentation for improving maintainability. A comparison between \eqref{SS_OP_Asym} with \eqref{SA_OP_Asym} further indicates that SA yields a substantially faster decay of OP than SS. This advantage arises because SA exploits multiple working segments simultaneously, which strengthens reliability and improves maintainability under the segmented-waveguide architecture.
\end{remark}
\vspace{-5pt}
\subsubsection{Comparison with Conventional PASS}\label{Section: SA: Comparison with the Conventional PASS}
We now compare SA-based SWAN with conventional PASS. Since the OP under SA does not yield a tractable closed-form expression, either in its exact form or under the Gaussian approximation, the comparison therefore focuses on the PNR. 

From \eqref{PNR_Conventional_PASS_Further} and \eqref{PNR_SA_Basic}, the relative PNR gain of SA-based SWAN over conventional PASS is given by
\begin{align}\label{Performance_Gain_SA_PASS}
{\mathcal{A}}_{\rm{SA}}\triangleq\frac{{\mathcal{P}}_{\rm{SA}}^{+}}{{\mathcal{P}}_{\rm{M}}^{+}}=(1+D_x^2\varepsilon_{0})\left(1-\left(\frac{L^2\varepsilon_{0}}{L^2\varepsilon_{0}+1}\right)^M\right).
\end{align}
For a fixed service region width $D_x$, ${\mathcal{A}}_{\rm{SA}}$ increases monotonically with the segment number $M$, since $\frac{\partial {\mathcal{A}}_{\rm{SA}}}{\partial M}=\frac{1}{{\mathcal{P}}_{\rm{M}}^{+}}\frac{\partial {\mathcal{P}}_{\rm{SA}}^{+}}{\partial M}>0$. This trend is consistent with the behavior observed for ${\mathcal{A}}_{\rm{SS}}$. Besides, $\lim_{M\rightarrow\infty}{\mathcal{A}}_{\rm{SA}}=1+D_x^2\varepsilon_{0}$, which coincides with the asymptotic limit of ${\mathcal{A}}_{\rm{SS}}$ as $M\rightarrow\infty$. Although both gains converge to the same limit, their convergence rates differ significantly. The gap to the limit under SS satisfies $1+D_x^2\varepsilon_{0}-{\mathcal{A}}_{\rm{SS}}\simeq{\mathcal{O}}(M^{-2})$. Under SA, the gap $1+D_x^2\varepsilon_{0}-{\mathcal{A}}_{\rm{SA}}$ decays as ${\mathcal{O}}((D_x^2\varepsilon_{0})^{M}M^{-2M})$, and this decay rate is much faster than ${\mathcal{O}}(M^{-2})$ as $M\rightarrow\infty$. This behavior shows that SA reaches its asymptotic reliability gain substantially faster than SS. The result highlights the benefit of activating multiple segments in SWAN for enhanced maintainability and long-term reliability.

We next consider the case where the segment length $L$ is fixed while $M$ scales. Under this setting, it follows directly from \eqref{Performance_Gain_SA_PASS} that ${\mathcal{A}}_{\rm{SA}}$ increases monotonically with $M$, since $1+D_x^2\varepsilon_{0}=1+M^2L^2\varepsilon_{0}$ increases with $M$ while $\left(\frac{L^2\varepsilon_{0}}{L^2\varepsilon_{0}+1}\right)^M$ decreases with $M$. Moreover, the performance gain scales on the order of ${\mathcal{O}}(M^2L^2\varepsilon_{0})$ as $M$ increases. This scaling rate is faster than the ${\mathcal{O}}(\frac{\varepsilon_0 {L^2}}{\varepsilon_0 {L^2}+{1}}M^2)$ behavior observed under the SS protocol. The difference reflects the additional reliability benefit of exploiting all working segments under SA.

We next examine two limiting regimes of the F3R $\varepsilon_0$. When $\varepsilon_0\rightarrow0$, failures are rare and repairs are fast. In this regime, $\lim_{\varepsilon_0\rightarrow0}{\mathcal{A}}_{\rm{SA}}=1$, which matches the corresponding limit for ${\mathcal{A}}_{\rm{SS}}$. When $\varepsilon_0\rightarrow\infty$, failures are frequent and repairs are slow, so both ${\mathcal{P}}_{\rm{SA}}^{+}$ and ${\mathcal{P}}_{\rm{M}}^{+}$ approach zero, which indicates that persistent failures severely degrade the availability of all architectures. However, the relative performance gain satisfies
\begin{align}\label{SA_Asym_Gain_Environment}
\lim\nolimits_{\varepsilon_0\rightarrow\infty}{\mathcal{A}}_{\rm{SA}}=M^3,
\end{align}
where a detailed derivation is provided in Appendix \ref{Proof_Limitation_Basic}. 
\vspace{-5pt}
\begin{remark}\label{Remark_Performance_Gain_Environment}
This result demonstrates that, even under extremely unreliable conditions, segmentation enables SA-based SWAN to achieve an $M^3$-fold improvement in PNR relative to conventional PASS. This gain exceeds the $M^2$-fold improvement obtained under SS, and it further confirms the benefit of exploiting multiple segments simultaneously.
\end{remark}
\vspace{-5pt}
\section{Numerical Results}\label{Section_Numerical_Results}
This section presents numerical results to validate the analytical results. Unless stated otherwise, the system parameters are set as follows: carrier frequency $f_{\rm{c}} = 28$ GHz, effective refractive index $n_{\rm{eff}} = 1.4$, minimum inter-antenna spacing $\Delta = \frac{\lambda}{2}$, transmit power $P = 10$ dBm, and noise power $\sigma^2=-90$ dBm. The service region is modeled as a rectangular area centered at the origin, with side lengths $D_x$ and $D_y = 20$ m along the $x$- and $y$-axes, respectively, as illustrated in {\figurename} {\ref{Figure: SWAN_System_Model}}. The user is located at the origin, so $u_x=u_y=0$ m. The waveguide is deployed at height $d = 3$ m along the $x$-axis with fixed $y$-coordinate equal to zero. The $x$-coordinate of the feed point of the first waveguide segment is given by $\psi_{0}^{1}=-\frac{D_x}{2}$. The maintainability of SWAN is compared with that of conventional PASS, which uses a single monolithic waveguide deployed over the same region and uses the same feed-point location $\psi_0^{\rm{M}}=-\frac{D_x}{2}$. All results are averaged over $10^6$ independent realizations of the waveguide operational states, namely, \emph{working} or \emph{failed}.

\begin{figure}[!t]
\centering
    \subfigure[$D_x=50$ m.]
    {
        \includegraphics[width=0.45\textwidth]{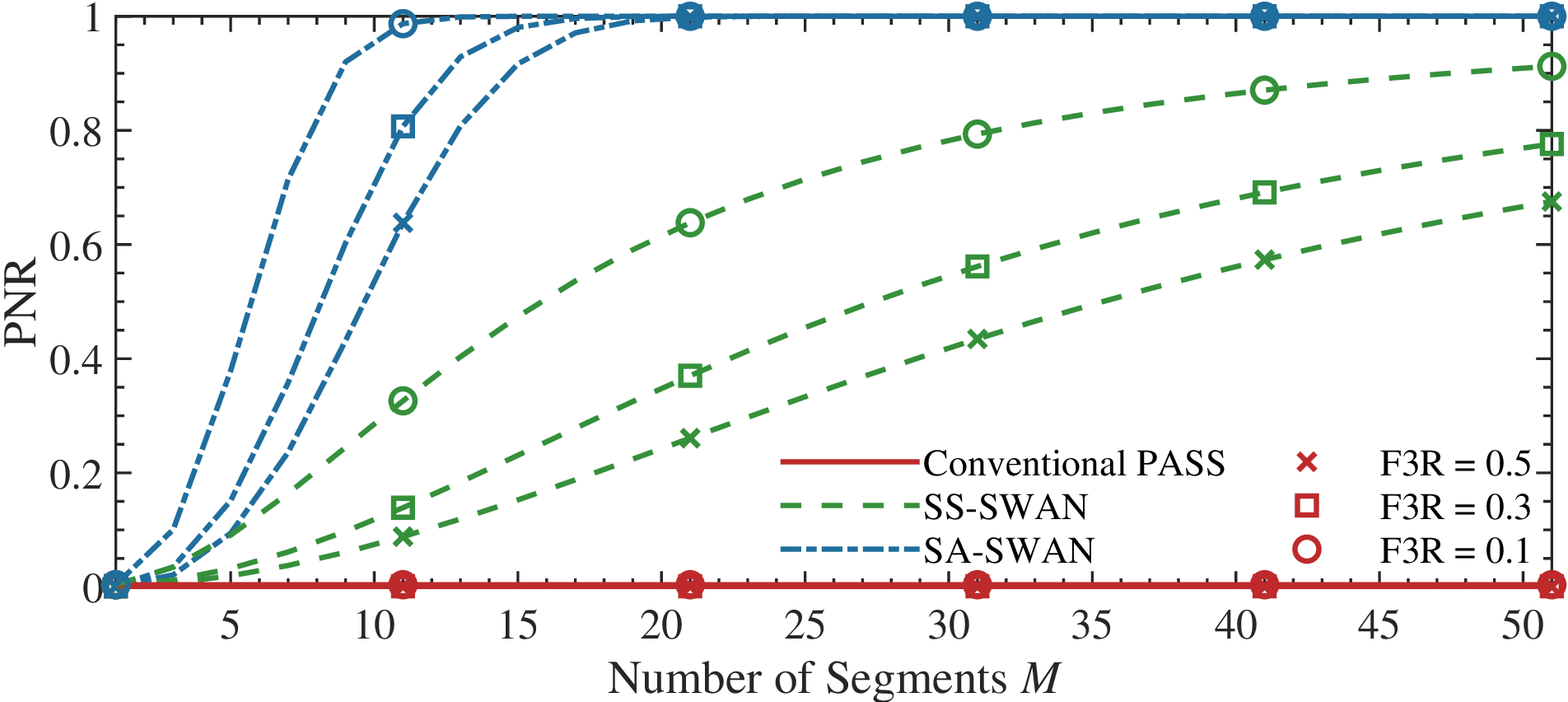}
	   \label{Figure_PNR_Number_Segment_Fig1}
    }
   \subfigure[$L=1$ m.]
    {
        \includegraphics[width=0.45\textwidth]{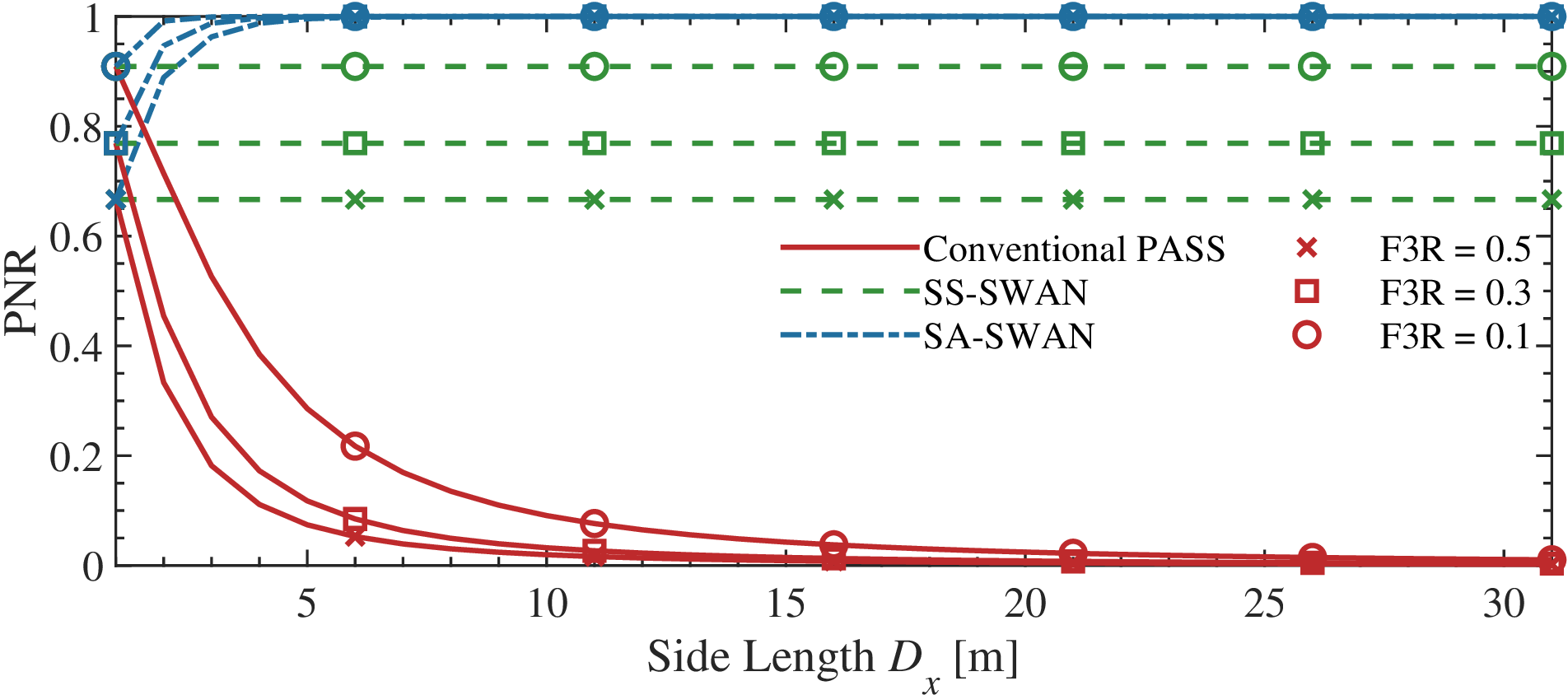}
	   \label{Figure_PNR_Number_Segment_Fig2}
    }
\caption{PNR comparison of SS-based SWAN, SA-based SWAN, and conventional PASS.}
\label{Figure: Figure_PNR_Number_Segment}
\vspace{-10pt}
\end{figure}

\subsection{Probability of Non-Zero Rate}
{\figurename} {\ref{Figure: Figure_PNR_Number_Segment}} compares the PNR of SS-based SWAN, SA-based SWAN, and conventional PASS in order to evaluate their maintainability. 

In {\figurename} {\ref{Figure_PNR_Number_Segment_Fig1}}, the PNR is plotted as a function of the number of segments $M$ for a fixed service region width $D_x=50$ m and several values of the F3R $\varepsilon_0$. As observed from {\figurename} {\ref{Figure_PNR_Number_Segment_Fig1}}, under all considered reliability conditions, i.e., for all examined values of $\varepsilon_0$, conventional PASS achieves a PNR close to zero. This result indicates poor maintainability and low long-term reliability. The underlying reason is that conventional PASS relies on a single long waveguide, which is more vulnerable to random failures and incurs higher repair cost and longer downtime. In contrast, both SS-based SWAN and SA-based SWAN achieve a significantly higher PNR than conventional PASS for all illustrated system configurations. Moreover, the performance improvement becomes more pronounced as the number of segments increases, which agrees with the analytical trends derived in Section {\ref{Section:Uplink SWAN}}. {\figurename} {\ref{Figure_PNR_Number_Segment_Fig1}} also shows that SA-based SWAN consistently outperforms SS-based SWAN. This advantage arises because SA can exploit multiple working segments simultaneously, and the system provides a non-zero rate whenever at least one segment is in the working state.

Based on the analytical results in \eqref{SS_PNR_Asym} and \eqref{SA_PNR_Asym}, the PNRs achieved by both SS and SA converge to one as the number of segments increases while $D_x$ is fixed. However, the convergence rate of SA-based SWAN is substantially faster than that of SS-based SWAN. This trend is clearly illustrated in {\figurename} {\ref{Figure_PNR_Number_Segment_Fig1}}, where the PNR of SA-based SWAN approaches one rapidly once the number of segments reaches roughly $20$. These results confirm that, from a maintainability perspective, increasing the number of segments significantly enhances system reliability for a fixed service region size. In practice, deploying an excessively large number of segments is constrained by hardware cost, control complexity, and physical implementation considerations. Nevertheless, the numerical results demonstrate that a moderate number of segments is sufficient to achieve near-optimal reliability performance. This observation highlights a favorable tradeoff between system complexity and maintainability efficiency.

{\figurename} {\ref{Figure_PNR_Number_Segment_Fig2}} further illustrates the PNR as a function of the service region width $D_x$ or, equivalently, the number of segments when the segment length $L$ is fixed. As expected, the PNR achieved by conventional PASS decreases monotonically with $D_x$, as maintaining a longer monolithic waveguide leads to a higher likelihood of random failures as well as longer repair durations. In contrast, the PNR achieved by SS-based SWAN remains constant as the number of segments increases. This result follows from the fact that SS relies solely on the operational state of the segment closest to the user. Consequently, its reliability performance depends only on the fixed segment length rather than on the overall service region width. This observation is consistent with the analytical result in \eqref{SS_PNR_Expression} and demonstrates that SS-based SWAN maintains stable maintainability even when serving a wider region. Compared with both conventional PASS and SS-based SWAN, SA-based SWAN achieves a PNR that increases monotonically with the number of segments. This advantage originates from the simultaneous utilization of multiple segments under the SA protocol. As more segments are deployed, the probability that at least one segment remains operational increases exponentially, as indicated by \eqref{PNR_SA_Basic}. As a result, the PNR achieved by SA rapidly converges to one as the number of segments grows. Taken together, the results in {\figurename} {\ref{Figure_PNR_Number_Segment_Fig2}} further validates the effectiveness of segmented waveguides for improving maintainability and reliability in PASS-based communications.

\begin{figure}[!t]
\centering
\includegraphics[width=0.45\textwidth]{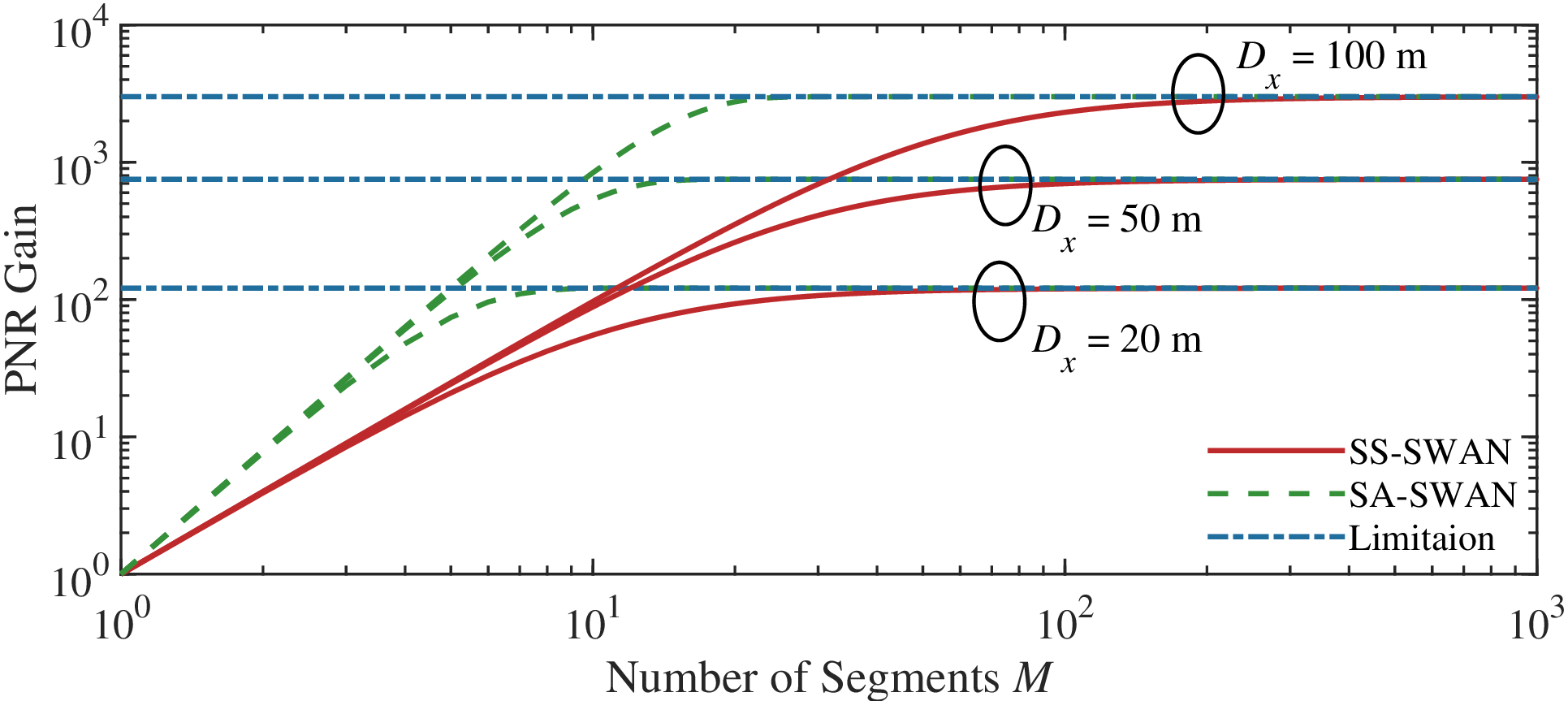}
\caption{PNR gain versus the segment number. $\varepsilon_0=0.3$.}
\label{Figure_Gain_Number_Segment_Fig1}
\vspace{-10pt}
\end{figure}

{\figurename} {\ref{Figure_Gain_Number_Segment_Fig1}} compares the PNR gains of SS-based and SA-based SWAN relative to conventional PASS, i.e., ${\mathcal{A}}_{\rm{SS}}=\frac{{\mathcal{P}}_{\rm{SS}}^{+}}{{\mathcal{P}}_{\rm{M}}^{+}}$ and ${\mathcal{A}}_{\rm{SA}}=\frac{{\mathcal{P}}_{\rm{SA}}^{+}}{{\mathcal{P}}_{\rm{M}}^{+}}$ under selected values of the service region width $D_x$. As shown in the figure, the PNR gains achieved by both SS and SA increase with the number of segments $M$, which is consistent with the trends observed in {\figurename} {\ref{Figure_PNR_Number_Segment_Fig1}}. Moreover, the PNR gain of SWAN over conventional PASS becomes more pronounced for larger service regions, i.e., for larger values of $D_x$. This observation supports the analytical conclusions in Remark \ref{Remark_Performance_Gain_Side_Length}. The underlying reason is that the reliability of the monolithic-waveguide PASS degrades rapidly with increasing $D_x$, whereas the reliability of SWAN improves substantially as the segment length decreases. {\figurename} {\ref{Figure_Gain_Number_Segment_Fig1}} also shows, as $M$ increases, both ${\mathcal{A}}_{\rm{SS}}$ and ${\mathcal{A}}_{\rm{SA}}$ converge to the same limiting value $1+D_x^2\varepsilon_{0}$, as derived in Sections \ref{Section: SS: Comparison with the Conventional PASS} and \ref{Section: SA: Comparison with the Conventional PASS}. However, ${\mathcal{A}}_{\rm{SA}}$ converges much faster than ${\mathcal{A}}_{\rm{SS}}$. This difference follows from the use of multiple segments under SA, which increases the probability that at least one segment remains operational. This observation is consistent with the analysis in Section \ref{Section: SA: Comparison with the Conventional PASS}.

\begin{figure}[!t]
\centering
\includegraphics[width=0.45\textwidth]{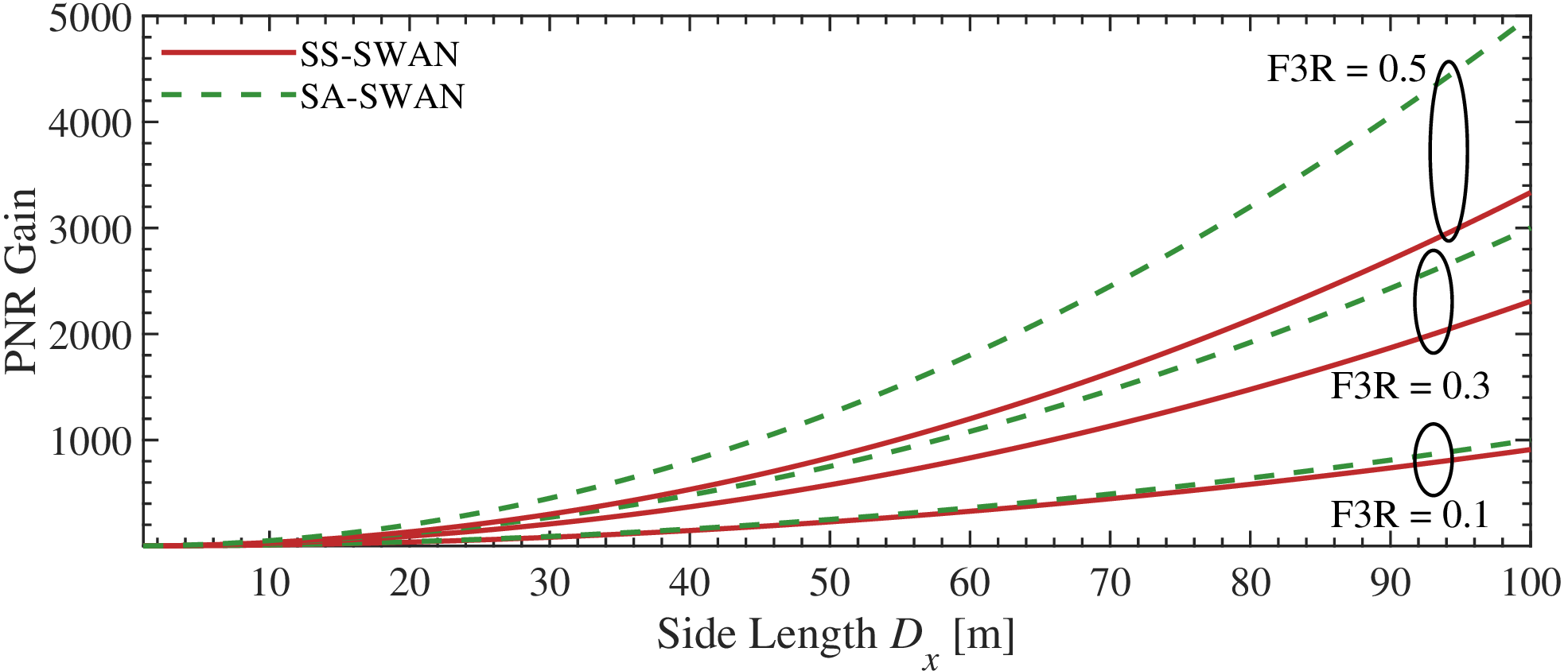}
\caption{PNR gain versus the side length. $L=1$ m.}
\label{Figure_Gain_Number_Segment_Fig2}
\vspace{-10pt}
\end{figure}

{\figurename} {\ref{Figure_Gain_Number_Segment_Fig2}} further depicts the PNR gain versus the segment number for a fixed segment length $L$ and selected values of the F3R $\varepsilon_0$. As shown in the figure, the PNR gains achieved by both SS and SA increase with the number of segments. The performance gap between the two schemes becomes more pronounced as the number of segments grows. This behavior is expected, since the SA protocol exploits all working segments, whereas the SS protocol relies only on the segment closest to the user. In addition, the growth rate of the PNR gain under SA exceeds that under SS. This observation aligns with the analytical results in Section \ref{Section: SA: Comparison with the Conventional PASS}, where ${\mathcal{A}}_{\rm{SA}}$ increases exponentially with $M$, while ${\mathcal{A}}_{\rm{SS}}$ grows only quadratically. {\figurename} {\ref{Figure_Gain_Number_Segment_Fig2}} also shows that, in highly reliable environments characterized by small values of $\varepsilon_0$, the performance gap between SA and SS becomes negligible. As $\varepsilon_0$ increases, this gap widens. This trend reflects the fact that employing multiple segments becomes increasingly beneficial for maintaining system reliability in less reliable environments.

\begin{figure}[!t]
\centering
\includegraphics[width=0.45\textwidth]{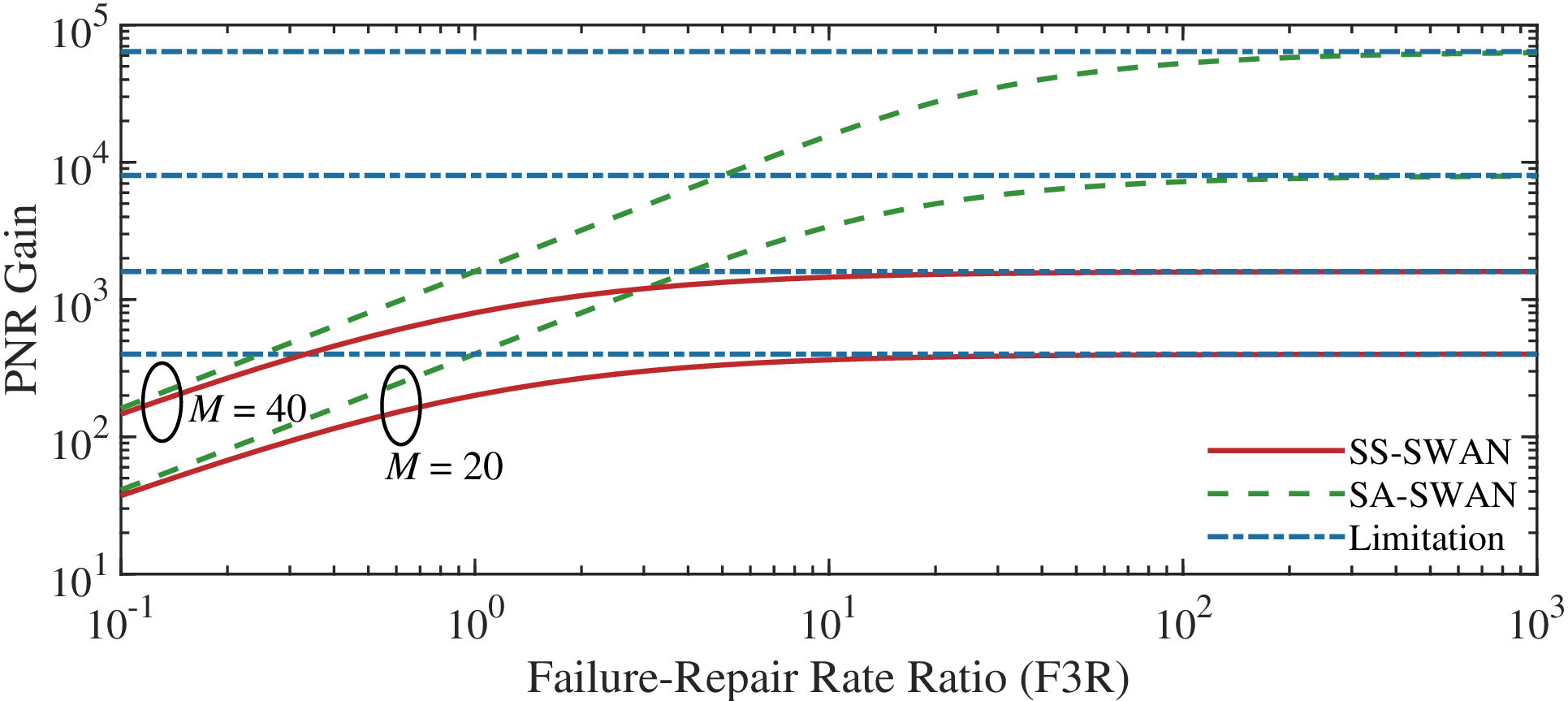}
\caption{PNR gain versus the F3R $\varepsilon_0$. $L=1$ m.}
\label{Figure_Gain_Number_Segment_Fig3}
\vspace{-10pt}
\end{figure}

{\figurename} {\ref{Figure_Gain_Number_Segment_Fig3}} illustrates the PNR gain as a function of the F3R $\varepsilon_0$ to examine the impact of environmental reliability on system performance. As shown in the figure, in highly reliable environments with small $\varepsilon_0$, the performance gains of both SS and SA over conventional PASS are marginal. In such cases, failures occur infrequently and repairs complete rapidly, which reduces the need for segmentation to improve maintainability. As $\varepsilon_0$ increases, the environment becomes more unreliable, and thus the effectiveness of the SWAN architecture becomes more pronounced, which leads to significantly higher PNR gains. In the larger-$\varepsilon_0$ regime, both ${\mathcal{A}}_{\rm{SA}}$ and ${\mathcal{A}}_{\rm{SS}}$ converge to their respective asymptotic limits, namely $M^2$ and $M^3$, which is consistent with the results stated in Remark \ref{Remark_Performance_Gain_Environment}.

\subsection{Outage Probability}
After examining the PNR, we turn to the OP, which is illustrated in {\figurename} {\ref{Figure: Figure_OP_Number_Segment_Fig}}.

{\figurename} {\ref{Figure_OP_Number_Segment_Fig1}} compares the OP achieved by SS-based SWAN, SA-based SWAN, and conventional PASS for selected values of the segment number $M$, with fixed service-region width $D_x$. For reference, the Gaussian-approximation-based upper bound for the OP of SA-based SWAN, derived from \eqref{OP_Upper_Bound_Gaussian_Approximation}, is also plotted. In addition, the OP achieved by SA without PA placement optimization is included. In this unoptimized case, the PA of each segment is placed at the center of the segment. Several observations can be drawn from {\figurename} {\ref{Figure_OP_Number_Segment_Fig1}}. First, as the number of segments increases, the OP achieved by both SS-based SWAN and optimized SA-based SWAN decreases. This trend confirms the effectiveness of the proposed antenna placement refinement method presented in Section \ref{SA_Antenna_Optimization_Uplink}. Second, the OP achieved by SA-based SWAN is substantially lower than that achieved by SS-based SWAN. Moreover, the decay rate of the OP under SA is much faster that that under SS, which is consistent with the behavior predicted by the Gaussian upper bound. These results demonstrate the clear advantage of simultaneously activating multiple segments under the SA protocol. A further observation is that the unoptimized SA baseline yields poor OP performance, and it remains close to that of conventional PASS. This observation occurs despite the use of multiple segments. The result indicates that segmentation alone can improve the PNR, but it does not guarantee a low OP when the PA locations do not ensure constructive signal combining. The comparison therefore highlights the importance of PA placement optimization in SWAN for achieving both strong communication performance and robust maintainability.

\begin{figure}[!t]
\centering
    \subfigure[$D_x=50$ m.]
    {
        \includegraphics[width=0.45\textwidth]{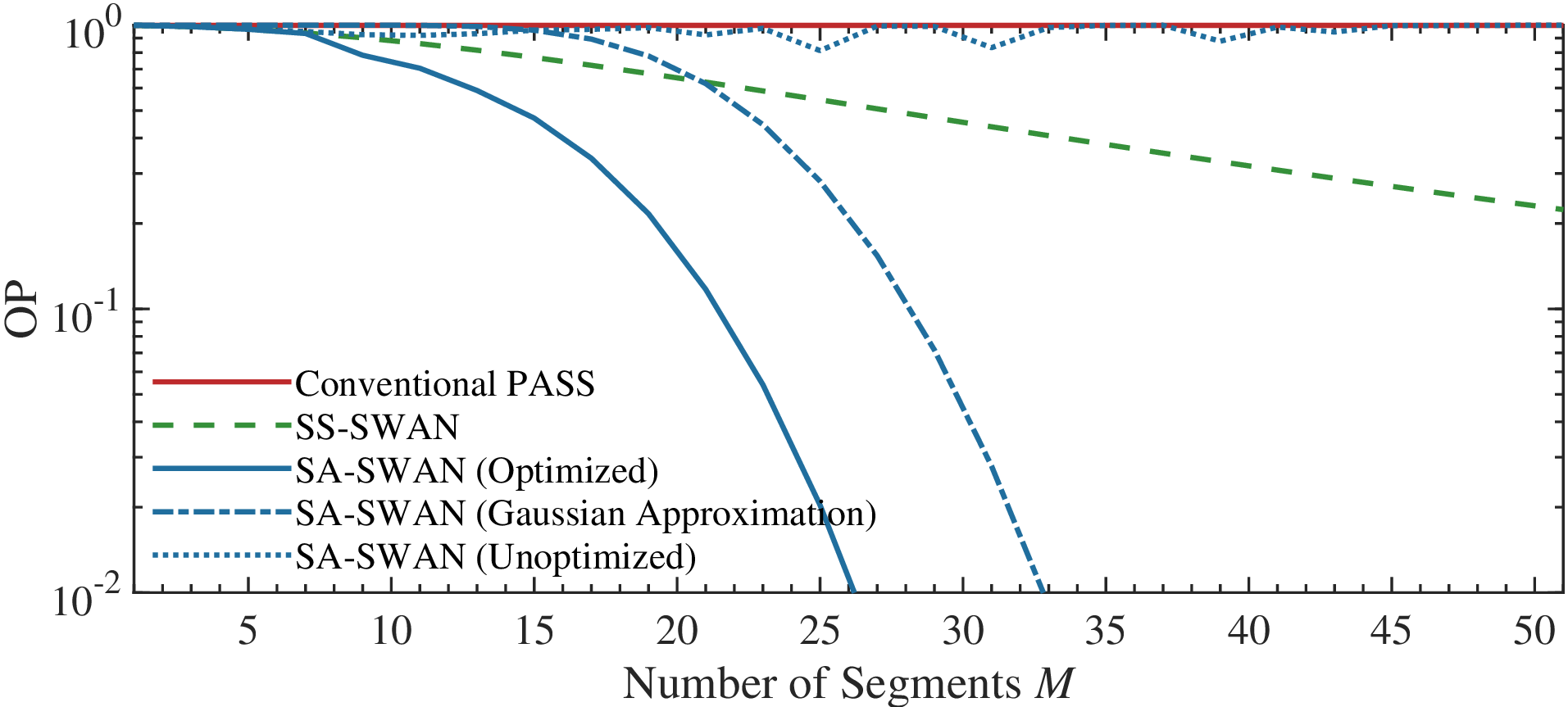}
	   \label{Figure_OP_Number_Segment_Fig1}
    }
   \subfigure[$L=1$ m.]
    {
        \includegraphics[width=0.45\textwidth]{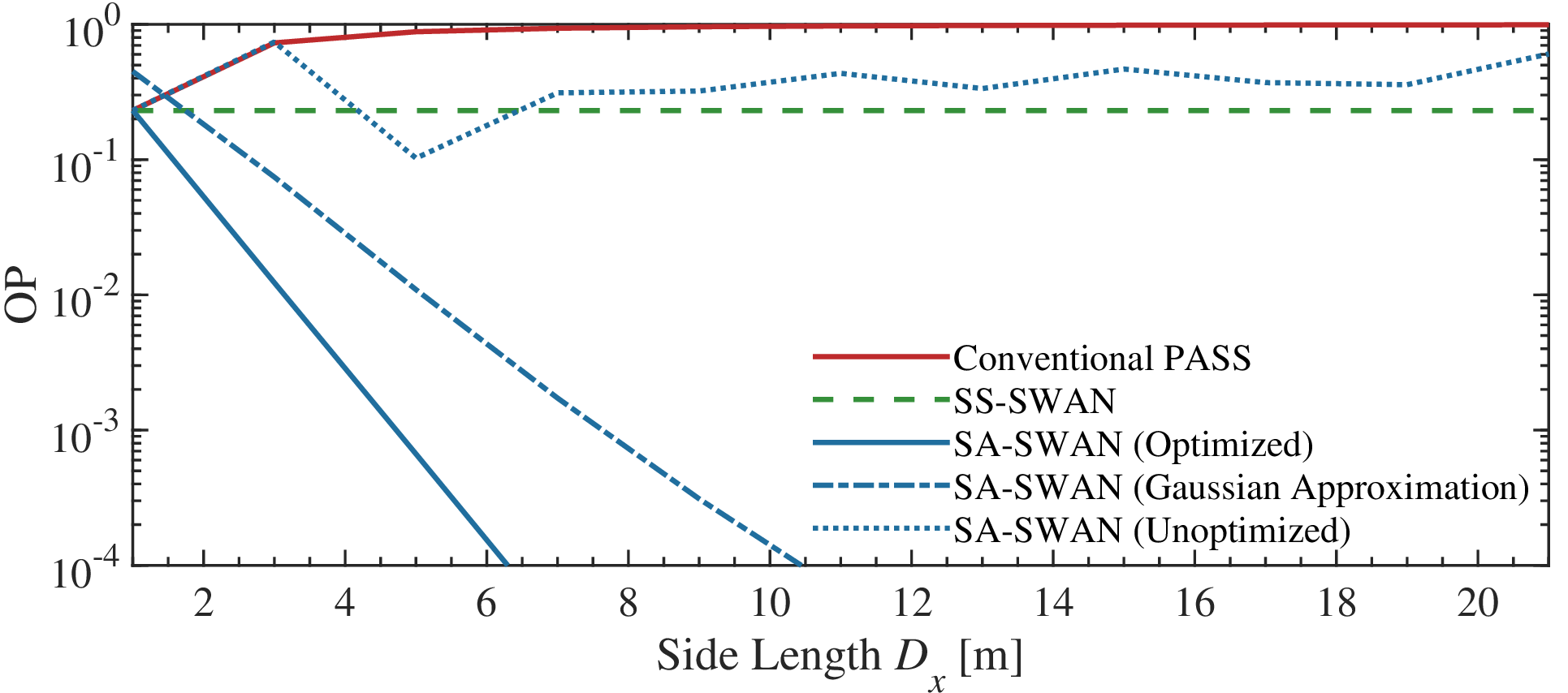}
	   \label{Figure_OP_Number_Segment_Fig2}
    }
\caption{OP comparison of SWAN and conventional PASS. $\varepsilon_0=0.3$ and $R_0=0.9\log_2(1+\gamma_{\rm{M}})$.}
\label{Figure: Figure_OP_Number_Segment_Fig}
\vspace{-10pt}
\end{figure}

{\figurename} {\ref{Figure_OP_Number_Segment_Fig2}} illustrates the OP as a function of the service region width $D_x$ for a fixed segment length $L$. As shown in the figure, the OP achieved by SS-based SWAN remains constant over $D_x$, since SS depends only on the operational state of the selected segment and this state is determined by the fixed segment length. In contrast, the OP of optimized SA-based SWAN decreases monotonically as $D_x$ increases, or equivalently, as the number of segments grows. This trend differs fundamentally from that of conventional PASS, whose OP increases with $D_x$. These observations are consistent with the trends shown in {\figurename} {\ref{Figure_PNR_Number_Segment_Fig2}}. {\figurename} {\ref{Figure_OP_Number_Segment_Fig2}} also shows that the unoptimized SA baseline maintains a high OP across all values of $D_x$, which again confirms that PA placement optimization is essential for realizing the full benefit of SA-based SWAN.
\section{Conclusion}\label{Section_Conclusion}
This article developed a CTMC-based analytical framework to characterize the maintainability and reliability of PASS under random waveguide failures and repairs. We analyzed the steady-state probabilities that a waveguide is operational or failed. By incorporating waveguide availability into the received SNR, we modeled the achievable rate of PASS as a random variable and quantified system maintainability and reliability using the PNR and OP. Using the proposed framework, we derived closed-form expressions for the PNR and OP of conventional PASS and SWAN, and characterized their scaling laws w.r.t. the service region size and the number of segments. The analysis demonstrated that SWAN provides substantially stronger maintainability than conventional PASS. Both analytical and numerical results showed that SWAN achieves a higher PNR and a lower OP than conventional PASS, since it relies on shorter segments that exhibit lower failure exposure and faster recovery than a monolithic waveguide under the adopted length-scaling model. The results also showed that SA-based SWAN outperforms SS-based SWAN, since SA can exploit multiple working segments simultaneously. This gain requires PA placement optimization to ensure constructive combining and to fully realize the benefits of SA. Collectively, these observations indicate that SWAN is a promising architecture for practical PASS deployment, particularly when system maintainability is a key design objective.
\begin{appendix}
\subsection{Proof of Lemma \ref{Lemma_Waveguide_CTMC_Transition_Probabilities}}\label{Proof_Lemma_Waveguide_CTMC_Transition_Probabilities}
Consider a small time interval $\Delta t$. If the waveguide is in the working state at time $t$, it remains in the working state with probability $1-\lambda_{\rm{M}}\Delta t+o(\Delta t)$ (see \eqref{Definition_Failture_Rate}) and transitions to the failed state with probability $\lambda_{\rm{M}}\Delta t+o(\Delta t)$. Conversely, if the waveguide is failed, it remains in that state with probability $1-\mu_{\rm{M}}\Delta t+o(\Delta t)$ and returns to the state of working with probability $\mu_{\rm{M}}\Delta t+o(\Delta t)$. Conditioned on the state at time $t$, the probability of being in the working state at time $t+\Delta t$, given that the system was working at time $0$, satisfies
\begin{equation}
\begin{split}
p_{11}(t+\Delta t)&=p_{11}(t)(1-\lambda_{\rm{M}}\Delta t)\\
&+p_{10}(t)(\mu_{\rm{M}}\Delta t)+o(\Delta t).
\end{split}
\end{equation}
Subtracting $p_{11}(t)$ from both sides, dividing by $\Delta t$, and taking the limit as $\Delta t\rightarrow0$ yields the following first-order differential equation:
\begin{align}
\frac{\rm{d}}{{\rm{d}}t}p_{11}(t)=-\lambda_{\rm{M}}p_{11}(t)+\mu_{\rm{M}}p_{10}(t).
\end{align}
Since $p_{10}(t)=1-p_{11}(t)$, this becomes
\begin{align}
\frac{\rm{d}}{{\rm{d}}t}p_{11}(t)=\mu_{\rm{M}}-(\lambda_{\rm{M}}+\mu_{\rm{M}})p_{11}(t),
\end{align}
subject to the initial condition $p_{11}(0)=1$. The solution to this first-order linear ordinary differential equation is given by
\begin{align}
p_{11}(t)=\frac{\mu_{\rm{M}}}{\lambda_{\rm{M}}+\mu_{\rm{M}}}+\frac{\lambda_{\rm{M}}}{\lambda_{\rm{M}}+\mu_{\rm{M}}}{\rm{e}}^{-(\lambda_{\rm{M}}+\mu_{\rm{M}})t}.
\end{align}
Using the identity $p_{10}(t)=1-p_{11}(t)$ immediately gives \eqref{Waveguide_CTMC_Transition_Probability2}. By symmetry, initializing the system in the failed state at $t=0$ leads to \eqref{Waveguide_CTMC_Transition_Probability3} and \eqref{Waveguide_CTMC_Transition_Probability4}, which completes the proof.
\subsection{Derivation of \eqref{SA_Asym_Gain_Environment}}\label{Proof_Limitation_Basic}
Recall that $\frac{L^2\varepsilon_{0}}{L^2\varepsilon_{0}+1}=1-\frac{1}{L^2\varepsilon_{0}+1}$. So, as $\varepsilon_{0}\rightarrow\infty$, 
\begin{align}
\left(1-\frac{1}{L^2\varepsilon_{0}+1}\right)^M\simeq1-\frac{M}{L^2\varepsilon_{0}+1}+{\mathcal{O}}\left(\frac{1}{\varepsilon_{0}^2}\right),
\end{align}
and hence
\begin{align}
1-\left(\frac{L^2\varepsilon_{0}}{L^2\varepsilon_{0}+1}\right)^M\simeq\frac{M}{L^2\varepsilon_{0}+1}+{\mathcal{O}}\left(\frac{1}{\varepsilon_{0}^2}\right).
\end{align}
Therefore, as $\varepsilon_0\rightarrow\infty$,
\begin{subequations}
\begin{align}
{\mathcal{A}}_{\rm{SA}}&\simeq
(1+D_x^2\varepsilon_{0})\left(\frac{M}{L^2\varepsilon_{0}+1}+{\mathcal{O}}\left(\frac{1}{\varepsilon_{0}^2}\right)\right)\\
&\simeq\frac{M(1+D_x^2\varepsilon_{0})}{L^2\varepsilon_{0}+1}+{\mathcal{O}}\left(\frac{1}{\varepsilon_{0}}\right).
\end{align}
\end{subequations}
Taking the limit $\varepsilon_{0}\rightarrow\infty$ yields $\lim\nolimits_{\varepsilon_0\rightarrow\infty}{\mathcal{A}}_{\rm{SA}}=\frac{MD_x^2\varepsilon_{0}}{L^2\varepsilon_{0}}=M\frac{D_x^2}{L^2}$. Using the identity $D_x=ML$, we obtain $\lim_{\varepsilon_0\rightarrow\infty}{\mathcal{A}}_{\rm{SA}}=M^3$.
\end{appendix}
\bibliographystyle{IEEEtran}
\bibliography{mybib}
\end{document}